\documentclass[11pt]{article}

\usepackage[top=50mm, bottom=50mm, left=50mm, right=50mm]{geometry}
\usepackage{lineno}
\usepackage{amssymb}
\usepackage{amsmath}
\usepackage{amsthm}
\usepackage{epsfig}
\usepackage{graphicx}
\usepackage{graphics}
\usepackage{float}
\usepackage{subfigure}
\usepackage{multirow}
\usepackage{color}
\usepackage{lineno}
\usepackage{fullpage}
\usepackage[normalem]{ulem} 
\usepackage{makeidx}
\usepackage{xspace}
\usepackage{wrapfig}
\makeindex

\newtheorem{theorem}{Theorem}

\newtheorem{Definition}{Definition}
\newtheorem{corollary}{Corollary}

\newtheorem{lemma}{Lemma}
\newtheorem{claim}{Claim}
\newtheorem{obs}{Observation}

\definecolor{darkred}{rgb}{1, 0.1, 0.3}
\definecolor{darkblue}{rgb}{0.1, 0.1, 1}
\definecolor{darkgreen}{rgb}{0,0.6,0.5}

\newcommand {\mm}[1] {\ifmmode{#1}\else{\mbox{\(#1\)}}\fi}
\newcommand{\denselist}{\itemsep 0pt\parsep=1pt\partopsep 0pt}

\newcommand{\etal}		{{et al.}\xspace}
\newcommand{\myeg}	{{e.g.}\xspace}

\newcommand{\eps}		{\varepsilon}
\newcommand{\myanceq}	{{\succeq}} %ancerster 
\newcommand{\myanc}	{{\succ}} %ancerster 
\newcommand{\mydesceq} {{\preceq}}
\newcommand{\mydesc}		{{\prec}}
\newcommand{\dgh}		{\delta_{\mathcal{GH}}} %GH distance
\newcommand{\optdI}	{\mu}
\newcommand{\myeps}	{\eps}
\newcommand{\mygood}		{{good}\xspace}
\newcommand{\myk}		{\tau} %{\kappa}
\newcommand{\mydelta}	{\delta}

\newcommand{\Img}		{{\mathrm{Im}}}
\newcommand{\myL}		{{L}}
\newcommand{\superlevel}		{{super-level}\xspace}
\newcommand{\myC}		{{\mathrm{C}}}
\newcommand{\setSL}	{{\mathcal{L}}}
\newcommand{\slone}[1] 	{{\mathrm{L}^{(1)}_{#1}}}
\newcommand{\sltwo}[1] 	{{\mathrm{L}^{(2)}_{#1}}}
\newcommand{\hatT}		{{\widehat{T}}}
\newcommand{\levelC}		{\mathrm{Ch}} %{{levelC}}
\newcommand{\LCA}		{{LCA}}%common ancester
\newcommand{\DPalg}	{{\sf DPgoodmap}}
\newcommand{\modifyDPalg}	{{\sf modified-DP}}
\newcommand{\myF}		{{\mathcal{F}}}
\newcommand{\mypartial}[1]		{{partial-{#1}-good}\xspace}

\newcommand{\degbound}		{{degree-bound}\xspace}
\newcommand{\mydepth}		{{depth}\xspace}

\newcommand{\Intdopt}		{{\mydelta^*}}
\newcommand{\criSet}		{{\Pi}}
\newcommand{\Tcal}			{\mathcal{T}}
\newcommand{\mycost}		{{\mathrm{cost}}}
\newcommand{\uTone}		{{|T_1^f|}}
\newcommand{\uTwo}		{{|T_2^g|}}
\newcommand{\uT}				{{|T|}}
\newcommand{\mywF}		{{w^F}}
\newcommand{\mykpara}	{{\myk_\mydelta}}
\newcommand{\newdepth}	{{depth}}
\newcommand{\newtau}	{{\widehat{\tau}}}

\newcommand{\sensiblepair}     {{sensible-pair}\xspace}
\newcommand{\Edgelistpair}      {{Edge-list pair}\xspace}
\newcommand{\edgelistpair}      {{edge-list pair}\xspace}
\newcommand{\Fnew}          {{F_{new}}}
\newcommand{\Fold}          {{F_{old}}}
\begin{document}

\title{FPT-Algorithms for computing Gromov-Hausdorff and interleaving distances between trees}
 
\author{
{Elena Farahbkhsh Touli\footnote{Department of Mathematics, Stockholm University}} \and {Yusu Wang\footnote{Department of Computer Science, Ohio State University, email: yusu@cse.ohio-state.edu}}
}

%\institute{Stockholm University, Stockholm, Sweden\\\email{elena.touli@math.su.se}\andthe Ohio State University Columbus, Ohio, U.S.A.\\\email{yusu@cse.ohio-state.edu}\\}\\

%\authorrunning{Mokhov, Sutcliffe and Voronkov}

% \title{FPT-Algorithms for computing Gromov-Hausdorff and interleaving distances between trees}
% \author{Elena Farahbkhsh Touli} \and \author{y}
%\date{}

%\begin{document}
\maketitle
%\linenumbers
\setcounter{page}{0}

\begin{abstract}
The Gromov-Hausdorff distance is a natural way to measure the distortion between two metric spaces. However, there has been only limited algorithmic development to compute or approximate this distance. We focus on computing the Gromov-Hausdorff distance between two metric trees. Roughly speaking, a metric tree is a metric space that can be realized by the shortest path metric on a tree. Any finite tree with positive edge weight can be viewed as a metric tree where the weight is treated as edge length and the metric is the induced shortest path metric in the tree. 
Previously, Agarwal et al. showed that even for trees with unit edge length, it is NP-hard to approximate the Gromov-Hausdorff distance between them within a factor of $3$. In this paper, we present a fixed-parameter tractable (FPT) algorithm that can approximate the Gromov-Hausdorff distance between two general metric trees within a \emph{multiplicative factor} of $14$. 

Interestingly, the development of our algorithm is made possible by a connection between the Gromov-Hausdorff distance for metric trees and the interleaving distance for the so-called merge trees. The merge trees arise in practice naturally as a simple yet meaningful topological summary (it is a variant of the Reeb graphs and contour trees), and are of independent interest. It turns out that an exact or approximation algorithm for the interleaving distance leads to an approximation algorithm for the Gromov-Hausdorff distance. One of the key contributions of our work is that we re-define the interleaving distance in a way that makes it easier to develop dynamic programming approaches to compute it. We then present a fixed-parameter tractable algorithm to compute the interleaving distance between two merge trees {\bf exactly}, which ultimately leads to an FPT-algorithm to approximate the Gromov-Hausdorff distance between two metric trees. This exact FPT-algorithm to compute the interleaving distance between merge trees is of interest itself, as it is known that it is NP-hard to approximate it within a factor of $3$, and previously the best known algorithm has an approximation factor of $O(\sqrt{n})$ even for trees with unit edge length. 
\end{abstract}

\newpage

\section{Introduction}
Given two metric spaces $(X, d_X)$ and $(Y, d_Y)$, a natural way to measure their distance is via the \emph{Gromov-Hausdorff distance} $\dgh(X, Y)$ between them \cite{GH-book}, which intuitively describes how much \emph{additive} distance distortion is needed to make the two metric spaces isometric. 

We are interested in computing the Gromov-Hausdorff distance between \emph{metric trees}. 
Roughly speaking, a metric tree $(X, d)$ is a geodesic-metric space that can be realized by the shortest path metric on a tree. 
Any finite tree $T = (V, E)$ with positive edge weights $w: E\to \mathbb{R}$ can be naturally viewed as a metric tree $\Tcal = (|T|, d)$: 
the space is the underlying space $|T|$ of $T$, each edge $e$ can be viewed as a segment of length $w(e)$, and the distance $d$ is the induced shortest path metric. See Figure \ref{fig:trees} (a) for an example. 
Metric trees occur commonly in practical applications: 
\myeg, a neuron cell has a tree morphology, and can be modeled as an embedded metric tree in $\mathbb{R}^3$. 
It also represents an important family of metric spaces that has for example attracted much attention in the literature of metric embedding and recovery of hierarchical structures, \myeg, \cite{agarwala1998approximability,alon2008ordinal,ailon2005fitting,badoiu2007approximation,chepoi2012constant, fellows2008parameterized,SWW17}. 

Unfortunately, it is shown in \cite{AFN18,Sch17} that it is not only NP-hard to compute the Gromov-Hausdorff distance between two trees, but also NP-hard to approximate it within a factor of $3$ even for trees with unit edge length. A polynomial-time approximation algorithm is given in \cite{AFN18}; however, the approximation factor is high: it is $O(\sqrt{n})$ even for unit-edge weight trees. 

Another family of tree structures that is of practical interest is the so-called \emph{merge tree}. 
Intuitively, a merge tree is a rooted tree $T$ associated with a real-valued function $f: T \to \mathbb{R}$ such that the function value is monotonically decreasing along any root-to-leaf path -- We can think of a merge tree to be a tree with height function associated to it where all nodes with degree $> 2$ are down-forks (merging nodes); see Figure \ref{fig:trees} (b).  
The merge tree is a loop-free variant of the so-called Reeb graph, which is a simple yet meaningful topological summary for a scalar field $g: X \to \mathbb{R}$ defined on a domain $X$, and has been widely used in many applications in graphics and visualization \myeg, \cite{BGSF08,GSBW11,HA03,Tie08}. 
Morozov \etal{} introduced the \emph{interleaving distance} to compare merge trees \cite{IDbMT}, based on a natural ``interleaving idea'' which has recently become fundamental in comparing various topological objects. 
Also, several distance measures have been proposed for the Reeb graphs \cite{bauer,BLM18,deSilva2016}. When applying them to merge trees, it turns out that two of these distance measures are equivalent to the interleaving distance. 
However, the same reduction in \cite{AFN18} to show the hardness of approximating the Gromov-Hausdorff distance can also be used to show that it is NP-hard to approximate the interleaving distance between two merge trees within a factor 3. 

%\paragraph{New work.}
\vspace*{0.08in}\noindent{\bf New work.}
Although the Gromov-Hausdorff distance is a natural way to measure the degree of near-isometry between metric spaces \cite{GH-book,ms-tcfii-05}, the algorithmic development for it has been very limited so far \cite{AFN18,bbk-eciid-06,dgh-props,Sch17}. 
In \cite{Sch17}, Schmiedl gave an FPT algorithm for approximating the Gromov-Hausdorff distance between two \emph{finite metrics}, where the approximation contains \emph{both an additive and multiplicative terms}; see more discussion in {\it Remarks} after Theorem \ref{thm:GHalg}. 
%\footnote{We also note that the time complexity of \cite{Sch17} contains terms $n^k$, where $k$ is the fixed parameter and can be rather large -- Indeed, $k$ is the cardinality of a $\eps$-net of the input metric space, and $\eps$ also appears as an additive approximation term for the algorithm. In contrast, the dependency of our algorithm on the parameter $\hat{\tau}$ is roughly $O(2^\hat{\tau})$}.
In this paper, we present the first FPT algorithm to approximate the Gromov-Hausdorff distance for metric trees \emph{within a constant multiplicative factor}. 
%Even for two discrete metric spaces induced by points in fixed dimensional Euclidean space, it is not yet clear how to approximate the Gromov-Hausdorff distance between them efficiently within a constant factor. 
%In this paper, we present the first FPT algorithm to approximate the Gromov-Hausdorff distance for non-trivial inputs (metric trees more precisely). 

Interestingly, the development of our approximation algorithm is made possible via a connection between the Gromov-Hausdorff distance between metric trees and the interleaving distance between certain merge trees (which has already been observed previously in \cite{AFN18}). 
This connection implies that any exact or approximation algorithm for the interleaving distance will lead to an approximation algorithm for the Gromov-Hausdorff distance for metric trees of similar time complexity. 
Hence we can focus on developing algorithms for the interleaving distance. 
The original interleaving distance definition requires considering a pair of maps between the two input merge trees and their interaction. 
One of the key insights of our work is that we can in fact develop an equivalent definition for the interleaving distance that relies on only \emph{a single map} from one tree to the other. 
This, together with the height functions equipped with merge trees (which give rises to natural ordering between points in the two trees), essentially allows us to develop a dynamic programming algorithm to check whether the interleaving distance between two merge trees is smaller than a given threshold or not: In particular in Section \ref{sec:decision}, we first give a simpler DP algorithm with slower time complexity to illustrate the main idea. We then show how we can modify this DP algorithm to improve the time complexity. 
Finally, we solve the optimization problem for computing the interleaving distance\footnote{We note that the final time complexity for the optimization problem presented in Theorem \ref{thm:optinterleaving-fast} is based on an argument by Kyle Fox. His argument improves our previous $n^4$ factor (as in Theorem \ref{thm:optinterleaving-slow}) by an almost $n^2$ factor, by performing a double-binary search, instead of a sequence search we originally used.} in Section \ref{sec:optimization}, which leads to a constant-factor (a multiplicative factor of $14$) approximation FPT algorithm for the Gromov-Hausdorff distance between metric trees. 
%Some open problems are given in Section \ref{sec:conclusion}. 

\begin{figure}[tbph]
\begin{center}
\begin{tabular}{ccccc}
\includegraphics[height=3cm]{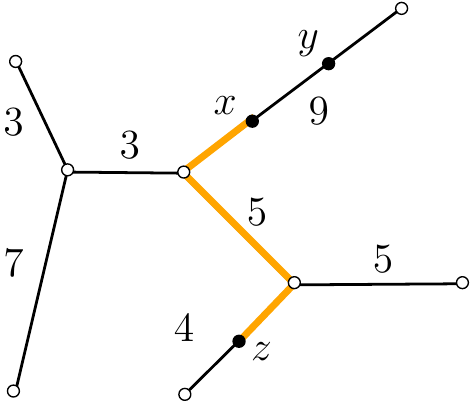} &\hspace*{0.1in} &
 \includegraphics[height=3.5cm]{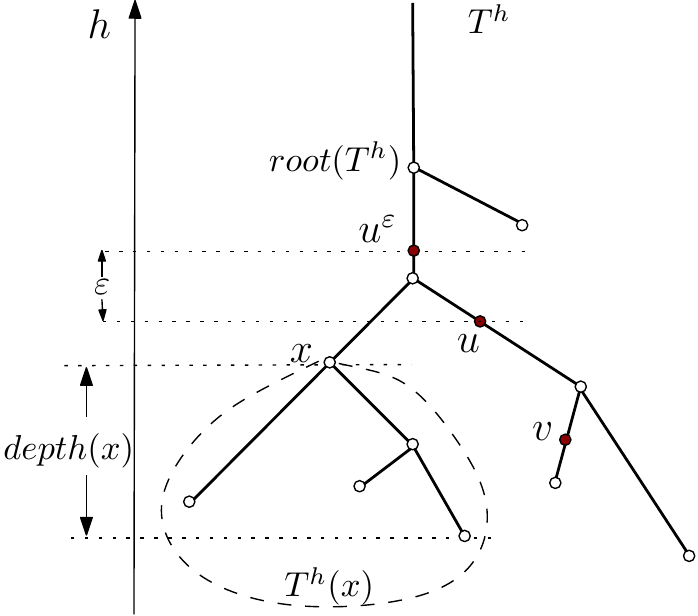} & \hspace*{0.1in}
& 
\includegraphics[height=3cm]{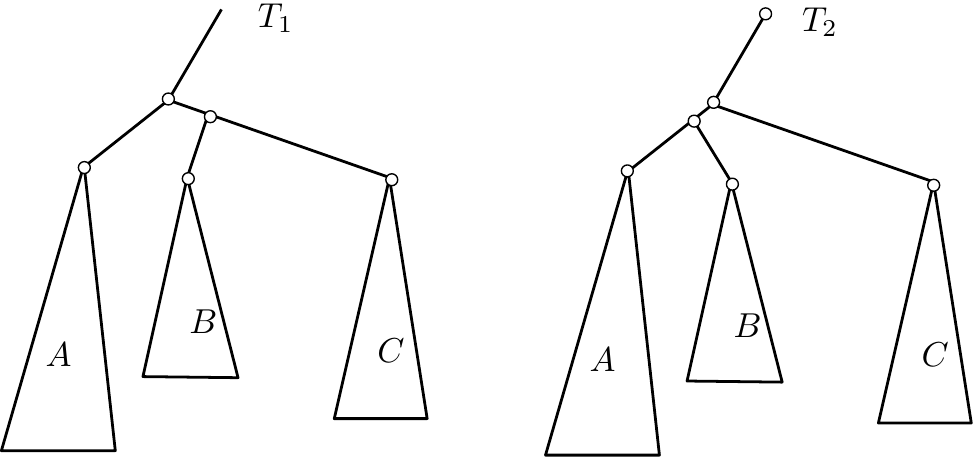} \\
(a) & & (b) & & (c) 
\end{tabular}
\end{center}
\vspace*{-0.3in}\caption{
(a) A metric tree $(T,d_T)$ with edge length marked. Tree nodes are white dots. 
%$d_T(x,y) = 3$, $d_T(x,z) = 3 + 5 + 2 = 10$ is the length of the thickened path $\pi(x,z)$. 
$d_T(x,z) = 3 + 5 + 2 = 10$ is the length of the thickened path $\pi(x,z)$.
(b) A merge tree $T^h$, with examples of $u \myanc v$, $u^\eps$, $T^h(x)$ and $\newdepth(x)$ marked. 
(c) Tree alignment distance between $T_1$ and $T_2$ arbitrarily large, while $\dgh(T_1,T_2)$ is roughly bounded by the pairwise distance difference which is small. 
\label{fig:trees}}
\end{figure}

%\paragraph{More on related work.}
%\paragraph{Remarks on other tree distances.} 
\noindent{\bf More on related work.~} 
There have been several tree distances proposed in the literature. 
Two most well-known ones are the tree edit and tree alignment distances \cite{B05}, primarily developed to compare labeled trees. % (ordered or unordered).   
%Perhaps the two most well-known tree distances are the tree edit distance and the tree alignment distance \cite{B05}, which are developed mostly to compare labeled trees (ordered or unordered).   
Unfortunately, both distances are MAX SNP-hard to compute for un-ordered trees \cite{JWZ95,ZJ94}. 
For tree edit distance, it is MAX SNP-hard even for trees with bounded degree. 
For the tree alignment distance, it can be computed in polynomial time for trees with bounded degree. 
However the tree alignment distance requires that parent-child relation to be strictly preserved, and thus the small local configuration change shown in Figure \ref{fig:trees} (c) will incur a large tree alignment distance.  
%and thus does not allow \myeg, local configuration change shown in Figure \ref{fig:trees} (c). 

We will not survey the large literature in metric embedding which typically minimizes the metric distortion in a \emph{mulplicative} manner. However, we mention the work of Hall and Papadimitriou \cite{hp-ad-05}, where, given two equal-sized point sets, they propose to find the best \emph{bijection} under which the \emph{additive distortion} is minimized. They show it is NP-hard to approximate this optimal additive distortion within a factor of 3 even for points in $\mathbb{R}^3$. In contrast, the Gromov-Hausdorff distance is also \emph{additive}, but allows for many-to-many correspondence (instead of bijection) between points from two input metric spaces. 
We also note that our metric trees consist of all points in the underlying space of input trees (i.e, including points in the interior of a tree edge). 
This makes the distance robust against adding extra nodes and short ``hairs'' (branches). 
Nevertheless, we can also consider discrete metric trees, where we only aim to align nodes of input trees (instead of all points in the underlying space of the trees). 
Our algorithms hold for the discrete case as well.

\section{Preliminaries}
\label{sec:notations}

\paragraph{Metric space, metric trees.}
A metric space is a pair $(X, d)$ where $X$ is a set and $d: X \times X \to \mathbb{R}_{\ge 0}$ satisfies: (i) for any $x, y\in X$, $d(x,y) \ge 0$ and $d(x,y)=0$ holds only when $x = y$; (ii) $d(x,y) = d(y,x)$, and (iii) for any $x,y, z$, $d(x,z) \le d(x,y)  + d(y,z)$. We call $d$ a metric on the space $X$. 
A metric space $(X, d)$ is a finite metric tree if it is a length metric space\footnote{$(X,d)$ is a length metric space if $d$ is the same as the shortest path (i.e, intrinsic) metric it induces on $X$.} and $X$ is homeomorphic to the underlying space $|T|$ of some finite tree $T=(V,E)$. 

Equivalently, suppose we are given a finite tree $T=(V, E)$ where each edge $e\in E$ has a positive weight $\ell(e) > 0$. 
View the underlying space $|e|$ of $e$ as a segment with length $\ell(e)$ (i.e, it is isometric to $[0, \ell(e)]$), and we can thus define the distance $d_T(x,y)$ between any two points $x, y \in |e|$ as the length of the sub-segment $e[x, y]$. 
The underlying space $|T|$ of $T$ is the union of all these segments (and thus includes points in the interior of each edge as well). 
For any $x, z\in |T|$, there is a unique simple path $\pi(x,z) \subset |T|$ connecting them. The (shortest path) distance $d_T(x,z)$ equal to the length of this path, which is simply the sum of the lengths of the restrictions of this path to edges in $T$. 
See Figure \ref{fig:trees} (a). 
The space $|T|$ equipped with $d_T$ is a metric tree $(|T|, d_T)$.

%\paragraph{Some notation conventions.}
Given a tree $T=(V, E)$, we use the term \emph{tree nodes} to refer to points in $V$, and an arbitrary $x\in |T|$ potentially from the interior of some tree edge is referred to as a \emph{point}. 
Given $T$, we also use $V(T)$ and $E(T)$ to denote its node-set and edge-set, respectively. 
To emphasize the combinatorial structure behind a metric tree, in the paper we will write a metric tree $(T, d_T)$, with the understanding that the space is in fact the underlying space $|T|$ of $T$. 
%We may also sometimes simply refer to a metric tree $(T, d_T)$ (resp. a metric space $(X, d_X)$) as $T$ (resp. $X$) when the choice of its metric is clear from the context. 
%
%Similarly, later we often need to consider continuous maps of the form $f: |T| \to \mathbb{R}$ or $\alpha: |T_1| \to |T_2|$, to simplify the presentation, we write $f: T\to \mathbb{R}$ and $\alpha: T_1 \to T_2$ respectively. 
%\yusu{Maybe not for these. Using a macro now.}
%If we ever wish to refer to only the node set of $T$, then we will use $V(T)$, e.g, $\hat{f}: V(T) \to \mathbb{R}$. 

Note that if we restrict this metric space to only the tree nodes, we obtain a \emph{discrete metric tree $(V(T), d_T)$}, and the distance between two tree nodes is simply the standard shortest path distance between them in a weighted graph (tree $T$ in this case). 
Our algorithms developed in this paper can be made to work for the discrete metric trees as well.

\paragraph{Gromov-Hausdorff distance.}
Given two metric spaces $\mathcal{X} = (X, d_X)$ and $\mathcal{Y} = (Y, d_Y)$, a \emph{correspondence} between them is a relation $C: X \times Y$ whose projection on $X$ and on $Y$ are both surjective; i.e, for any $x\in X$, there is at least one $(x, y) \in C$, and for any $y' \in Y$, there is at least one $(x', y') \in C$. 
If $(x,y)\in C$, then we say $y$ (resp. $x$) is a pairing partner for $x$ (resp. $y$); note that $x$ (resp. $y$) could have multiple pairing partners. 
The cost of this correspondence is defined as:
$$\mycost(C) = \max_{(x,y), (x', y') \in C} | d_X(x, x') - d_Y (y, y') |,$$
which measures the maximum metric distortion (difference in pairwise distances) under this correspondence. 
The \emph{Gromov-Hausdorff distance} between them is:
$$\dgh(\mathcal{X}, \mathcal{Y}) = \frac{1}{2} \inf_{C \in \Pi(X, Y)} \mycost(C), ~~~
\text{where}~\Pi(X,Y)=\text{set of correspondences between}~X~\text{and}~Y.$$
%where $\Pi(X,Y)$ is the set of correspondences between $X$ and $Y$. 

\paragraph{Merge trees.}
A \emph{merge tree} is a pair $(T, h)$ where $T$ is a rooted tree, and the continuous function $h: \uT \to \mathbb{R}$ is \emph{monotone} in the sense the value of $h$ is decreasing along any root-to-leaf path. 
See Figure \ref{fig:trees} (b) for an example. For simplicity, we often write the merge tree as $T^h$, and refer to $h$ as  the \emph{height function}, and $h(x)$ \emph{the height} of a point $x\in |T|$. 
The merge tree is an natural object: \myeg, it can be used to model a hierarchical clustering tree, where the height of a tree node indicates the parameter when the cluster (corresponding to the subtree  rooted at this node) is formed. 
It also arises as a simple topological summary of a scalar function $\tilde{h}: M \to \mathbb{R}$ on a domain $M$, which tracks the connected component information of the sub-level sets $\tilde{h}^{-1}(-\infty, a]$ as $a\in \mathbb{R}$ increases.

To define the interleaving distance, 
we modify a merge tree $T^h$ slightly by extending a ray from $root(T^h)$ upwards with function value $h$ goes to $+\infty$. 
All merge trees from now on refer to this modified version.  
Given a merge tree $T^h$ and a point $x\in |T|$, $T^h(x)$ is the subtree of $T^h$ rooted at $x$, and the \emph{\mydepth{} of $x$ (or of $T^h(x)$)}, denoted by $\newdepth(x)$, is the largest function value difference between $x$ and any node in its subtree; that is, the height of the entire subtree $T^h(x)$ w.r.t. function $h$. 
Given any two points $u, v \in |T|$, we use $u \myanceq v$ to denote that $u$ is an ancestor of $v$; $u \myanc v$ if $u$ is an ancestor of $v$ and $u \neq v$. 
Similarly, $v \mydesceq u$ means that $v$ is a descendant of $u$. Also, the degree of a node in a merge tree is defined as the downward degree of the node. 
We use $LCA(u,v)$ to represent the lowest common ancestor of $u$ and $v$ in $|T|$.  
For any non-negative value $\eps \ge 0$, $u^\eps$ represents the unique ancestor of $u$ in $T$ such that $h(u^\eps) - h(u) = \eps$. See Figure \ref{fig:trees} (b).

\paragraph{Interleaving distance.}
We now define the interleaving distance between two merge trees $T_1^f$ and $T_2^g$, associated with functions $f: \uTone \to \mathbb{R}$ and $g: \uTwo \to \mathbb{R}$, respectively. 

\begin{Definition}[$\eps$-Compatible maps \cite{IDbMT}] \label{def:compatible}
A pair of continuous maps $\alpha: \uTone \rightarrow \uTwo$ and $\beta: \uTwo \rightarrow \uTone$ is \emph{$\eps$-compatible w.r.t $T_1^f$ and $T_2^g$} if the following four conditions hold: 
\begin{align*}
{\sf (C1)}.~   g(\alpha(u))=f(u)+\eps ~~ \text{and}~~  &{\sf (C2)}. ~  \beta \circ\alpha(u)=u^{2\eps}   ~~\text{for any}~u \in \uTone; %\label{eq11}
\\
 {\sf (C3)}.~ f(\beta(w))=g(w)+\eps~~   \text{and}~~ &{\sf (C4)}.~ \alpha \circ \beta(w)=w^{2\eps} ~~\text{for any}~w \in \uTwo . %\label{eq12}
\end{align*}
\end{Definition}
To provide some intuition for this definition, note that if $\eps = 0$, then $\alpha = \beta^{-1}$: In this case, the two trees $T_1$ and $T_2$ are not only isomorphic, but also the function values associated to them are preserved under the isomorphism. 
In general for $\eps > 0$, this quantity measures how far a pair of maps are away from forming a function-preserving isomorphism between $T_1^f$ and $T_2^g$. In particular, $\beta$ is no longer the inverse of $\alpha$. However, the two maps relate to each other in the sense that if we send a point $u \in \uTone$ to $\uTwo$ through $\alpha: \uTone \to \uTwo$, then bring it back via $\beta: \uTwo \to \uTone$, we come back at an ancestor of $u$ in $\uTone$ (i.e, property {\sf (C2)}). This ancestor must be at height $f(u) + 2\eps$ due to properties {\sf (C1)} and {\sf (C3)}. 

\begin{Definition}[Interleaving distance \cite{IDbMT}] \label{def:interleaving} 
The \emph{interleaving distance} between two merge trees $T_1^f$ and $T_2^g$ is defined as:
\begin{align*}
d_I(T_1^{f},T_2^{f}) & = \inf\{~\eps ~\mid~ \text{there exist a pair of}~\eps\text{-compatible maps w.r.t}~ T_1^f ~\text{and}~T_2^g \}.
\end{align*}
\end{Definition}

Interestingly, it is shown in \cite{AFN18} that the Gromov-Hausdorff distance between two metric trees is related to the interleaving distance between two specific merge trees. 
\begin{claim}[\cite{AFN18}] \label{claim:GHandInterleaving}
Given two metric trees $\mathcal{T}_1=(T_1,d_1)$ and $\mathcal{T}_2=(T_2,d_2)$ with node sets $V_1 = V(T_1)$ and $V_2 = V(T_2)$, respectively, let $f_u: |T_1| \to \mathbb{R}$ (resp. $g_w: |T_2| \to \mathbb{R}$) denote the geodesic distance function to the base point $u \in V_1$ (resp. $v \in V_2$) defined by $f_u (x) = -d_1(x, u)$ for any $x\in |T_1|$ (resp. $g_w(y) = -d_2(y, w)$ for any $y\in |T_2|$). 
Set $\optdI :=  \min_{u\in V_1, w\in V_2} d_I(T_1^{f_u},T_2^{g_w})$. We then have that 
$$ \frac{\optdI}{14} \leq  \dgh(\mathcal{T}_1,\mathcal{T}_2)\leq  2\optdI. $$
\end{claim}

Note that to compute the quantity $\mu$, we only need to check all pairs of \emph{tree nodes} of $T_1$ and $T_2$, instead of all pairs of points from $|T_1|$ and $|T_2|$. 

We say a quantity $A$ is a \emph{$c$-approximation} for a quantity $B$ if 
$\frac{A}{c} \le B \le c A$; obviously, $c \ge 1$ and $c = 1$ means that $A = B$. 
The above claim immediately suggests the following: 
\begin{corollary} 
If there is an algorithm to $c$-approximate the interleaving distance between any two merge trees in $T(n)$ time, where $n$ is the total complexity of input trees, then there is an algorithm to $O(c)$-approximate the Gromov-Hausdorff distance between two metric trees in $n^2 T(n)$ time. 
\end{corollary}

In the remainder of this paper, we will focus on developing an algorithm to compute the interleaving distance between two merge trees $T_1^f$ and $T_2^g$. 
In particular, in Section \ref{sec:epsgood} we will first show an equivalent definition for interleaving distance, which has a nice structure that helps us to develop a fixed-parameter tractable algorithm for the decision problem of ``Is 
$d_I(T_1^f, T_2^g) \ge \myeps$?'' in Section \ref{sec:decision}. 
We show how this ultimately leads to FPT algorithms to compute the interleaving distance {\bf exactly} and to {\bf approximate} the Gromov-Hausdorff distance in Section \ref{sec:optimization}. 

\section{A New Definition for Interleaving Distance}
\label{sec:epsgood}

Given two merge trees $T_1^f$ and $T_2^g$ and $\mydelta > 0$, to answer the question  ``Is $d_I(T_1^f, T_2^g) \le \mydelta$?", a natural idea
is to scan the two trees bottom up w.r.t the ``height'' values (i.e, $f$ and $g$), while checking for possible $\myeps$-compatible maps between the partial forests of $T_1^f$ and $T_2^g$ already scanned. 
However, the interaction between the pair maps $\alpha$ and $\beta$ makes it complicated to maintain potential maps. 
We now show that in fact, we only need to check for the existence of a \emph{single} map from $T_1^f$ to $T_2^g$, which we will call the $\myeps$-\mygood{} map. 
We believe that this result is of independent interest. 

\begin{Definition}[$\eps$-\mygood{} map] \label{def:epsgood} %d1
A continuous map $\alpha: \uTone \rightarrow \uTwo$ is ${\eps}$-\mygood{} if and only if:
\begin{description}\denselist
\item[{\sf (P1)}] for any $u \in \uTone$, we have $g(\alpha(u))=f(u)+\eps$; 
\item[{\sf (P2)}] if $\alpha(u_1)\succeq\alpha(u_2)$, then we have $u_1^{2\eps}\succeq u_2^{2\eps}$, (note $u_1 \succeq u_2$ may not be true); and
\item[{\sf (P3)}] if $w\in \uTwo \setminus \Img(\alpha)$,
then we have $|g(\mywF)-g(w)|\leq 2\eps$, where $\Img(\alpha) \subseteq \uTwo$ is the image of $\alpha$, and $\mywF$ is the lowest ancestor of $w$ in $\Img(\alpha)$.
\end{description}
\end{Definition}

A map $\rho: |T_1^{h_1}| \to |T_2^{h_2}|$ between two arbitrary merge trees $T_1^{h_1}$ and $T_2^{h_2}$ is \emph{monotone} if for any $u\in |T_1^{h_1}|$, we have that $h_2( \rho(u)) \ge h_1(u)$. In other word, $\rho$ carries any point $u$ from $T_1^{h_1}$ to a point higher than it in $T_2^{h_2}$. 
If $\rho$ is continuous, then it will map an ancestor of $u$ in $T_1^{h_1}$ to an ancestor of $\rho(u)$ in $T_2^{h_2}$ as stated below (but the converse is not necessarily true):  
\begin{obs}\label{obs:ancester}
Given a \emph{continuous and monotone} map $\rho: |T_1^{h_1}| \to |T_2^{h_2}|$ between two merge trees $T_1^{h_1}$ and $T_2^{h_2}$, we have that if $u_1\succeq u_2$ in $T_1^{h_1}$, then $\rho(u_1)\succeq\rho(u_2)$ in $T_2^{h_2}$. 

This implies that if $w = \rho(u)$ for $u\in |T_1^{h_1}|$, 
then $\rho$ maps the subtree $T_1^{h_1}(u)$ rooted at $u$ into the subtree $T_2^{h_2}(w)$ rooted at $w$. 
This also implies that if $w\notin \Img(\rho)$, neither does any of its descendant. 
\end{obs}
Note that an $\eps$-\mygood{} map, or a pair of $\eps$-compatible maps, are all monotone and continuous. Hence the above observations are applicable to all these maps.  

The main result of this section is as follows. Its proof is in Appendix \ref{appendix:thm:epsgoodtwo}. 

\begin{theorem}\label{thm:epsgoodtwo}
Given any two merge trees $T_1^f$ and $T_2^g$, then $d_I(T_1^f, T_2^g)\le \eps$ if and only if there exists an ${\eps}$-\mygood{} map ${\alpha}:\uTone \rightarrow \uTwo$. 
\end{theorem}
% The theorem follows from Lemma \ref{lem:epsgoodone} and \ref{lem:epsgoodtwo} below. We put the relatively straightforward proof of Lemma \ref{lem:epsgoodone} and \ref{lem:epsgoodtwo} in Appendix \ref{appendix:lem:epsgoodone} and \ref{appendix:betacontinuous} respectively. 
% \begin{lemma}\label{lem:epsgoodone}
% If $d_I(T_1^f, T_2^g) \le \eps$, then there exists an ${\eps}$-\mygood{} map ${\alpha}:\uTone \rightarrow \uTwo$.
% \end{lemma}

% \begin{lemma}\label{lem:epsgoodtwo}
% If there is an $\eps$-\mygood{} map $\alpha: \uTone \to \uTwo$, then $d_I(T_1^f, T_2^g) \le \eps$. 
% \end{lemma}

\section{Decision Problem for Interleaving Distance}
\label{sec:decision}

\begin{figure}[htbp]
\vspace*{-0.1in}
\begin{center}
\begin{tabular}{ccc}
\includegraphics[width=4cm]{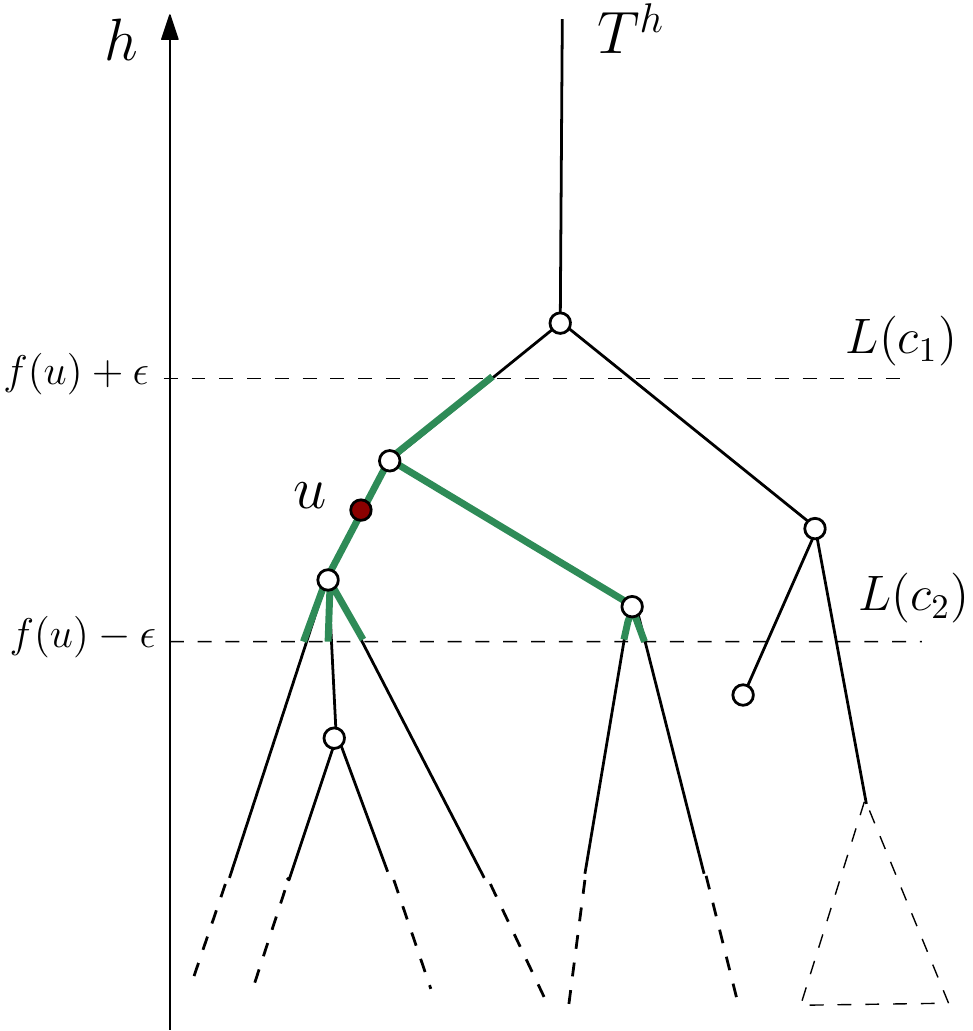} & \hspace*{0.2in}&
\includegraphics[height=4.5cm]{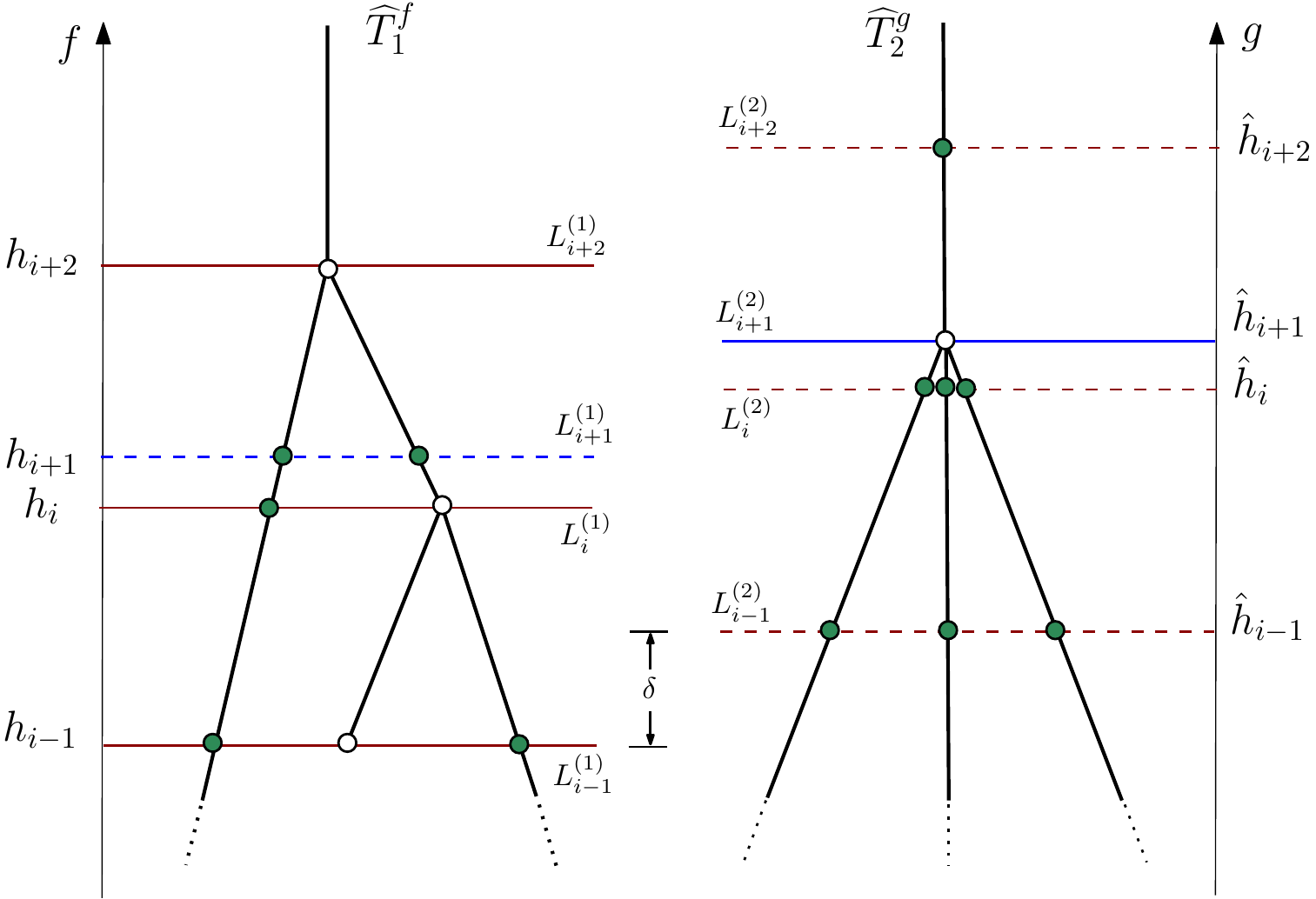} \\
(a) & & (b) 
\end{tabular}
\end{center}
\vspace*{-0.25in}
\caption{(a) Green component within the slab is $B_\eps(u, T_1^f)$. The sum of degrees for nodes within this $\eps$-ball is $13$. The $\eps$-degree bound $\myk_\eps(T_1^f, T_2^g)$ is the largest value of this sum for any $\eps$-ball in $T_1^f$ or in $T_2^g$. 
(b) White dots are tree nodes of $T_1^f$ and $T_2^g$. Green dots are newly augmented tree nodes in $\hatT_1^f$ and $\hatT_2^g$. 
\label{fig:augmentedtree}}
\end{figure}
%%%%%%%%%%%%%%%%%%%%%%%%%%%%%%%%%%
%\begin{wrapfigure}{r}{0.3\textwidth}
%%\begin{center}
%%\vspace*{-0.4in}
%\vspace*{-0.1in}
%\centering \includegraphics[width=0.2\textwidth]{./Figs/epsball}
%%\end{center}
%\vspace*{-0.2in}
%\caption{Green component within the slab is $B_\eps(u, T_1^f)$. 
%\label{fig:epsball}}
%\end{wrapfigure}
%%%%%%%%%%%%%%%%%%%%%%%%%%%%%%%%%% 
 In this section, given two merge trees $T_1^f$ and $T_2^g$ as well as a positive value $\mydelta > 0$, we aim to develop a fixed-parameter tractable algorithm for the decision problem ``Is $d_I(T_1^f, T_2^g) \le  \mydelta$?''. 
The specific parameter our algorithm uses is the following: 
Given a merge tree $T^h$ and any point $u\in T^h$, let $B_\eps(u; T^h)$ denote the $\eps$-ball 
$$B_\eps(u; T^h) = \{ x \in |T| \mid \forall y\in \pi_{T} (u, x), |h(y) - h(u)| \le \eps \},$$
where $\pi_{T}(u, x)$ is the unique path from $u$ to $x$ in $T^h$. 
In other words, $B_\eps(u; T^h)$ contains all points reachable from $u$ via a path whose function value is completely contained with the range $[f(u)-\eps, f(u)+\eps]$. 
See Figure \ref{fig:augmentedtree} (a) for an example: in particular, consider the restriction of $T^h$ within the height interval $[f(u)-\eps, f(u)+\eps]$. There could be multiple components within this slab, and $B_\eps(u; T^h)$ is the component containing $u$. 
\begin{description}\label{descript:parameter}
\item[Parameter $\myk_\mydelta$]: Let $\myk_\eps(T_1^f, T_2^g)$ denote  
the largest sum of degrees of all tree nodes contained in any $\myeps$-ball in $T_1^f$ or $T_2^g$, which we also refer to as the \emph{$\myeps$-\degbound{}} of $T_1^f$ and $T_2^g$. 
% $\myk_\eps(T_1^f, T_2^g)$: the numbers of leaves of the submerge tree inside of the $\epsilon$-ball. 
The parameter for our algorithm for the decision problem will be $\mykpara = \myk_\mydelta(T_1^f, T_2^g)$. 
\end{description}

\subsection{A Slower FPT-Algorithm}
\label{subsec:FPT}

\paragraph{Augmented trees. }
We now develop an algorithm for the decision problem  ``Is $d_I(T_1^f, T_2^g) \le \mydelta$?'' via a dynamic programming type approach. 
First, we will show that, even though a $\mydelta$-\mygood{} map is defined for all (infinite number of) points from $T_1^f$ and $T_2^g$, we can check for its existence by inspecting only a finite number of points from $T_1^f$ and $T_2^g$. 
In particular, we will augment the input merge trees $T_1^f$ and $T_2^g$ with extra tree nodes, and our algorithm later only needs to consider the discrete nodes in these augmented trees to answer the decision problem.

The set of points from tree $T_1^f$ or $T_2^g$ at a certain height value $c$ is called a \emph{level (at height $c$)}, denoted by $\myL(c)$. 
For example, in Figure \ref{fig:augmentedtree} (a), the level $L(c_1)$ for $c_1 = f(u)  + \eps$, contains 2 points, while $L(c_2)$ with $c_2 = f(u) - \eps$ contains 7 points. 
%See Figure \ref{fig:augmentedtree} (a).
The function value of a level $L$ is called its \emph{height}, denoted by $height(L)$; so $height(\myL(c)) = c$. 
\begin{Definition}[Critical-heights and Super-levels]\label{def:superlevels}
For the tree $T_1^f$, the \emph{set of critical-heights} $\myC_1$ 
consists of the function values of all tree nodes of $T_1^f$; similarly, define $\myC_2$ for $T_2^g$. That is, 
$$\myC_1 := \{ f(x) \mid x~ \text{is a tree node of}~ T_1^f \}; ~\text{and}~\myC_2  := \{ g(y) \mid y~ \text{is a tree node of}~ T_2^g \}. $$ 

\noindent The \emph{set of \superlevel{}s $\setSL_1$ w.r.t. $\mydelta$} for $T_1^f$ and the \emph{set of \superlevel{}s $\setSL_2 $} for $T_2^g$ are: 
\begin{align*}
\setSL_{1} &:= \{ \myL(c) \mid c \in \myC_1 \} \cup \{ \myL(c-\mydelta) \mid c \in \myC_2 \} ~\text{while} \\
\setSL_{2} &:= \{ \myL(c+\mydelta) \mid c \in \myC_1 \} \cup \{ \myL(c) \mid c \in \myC_2 \}. 
\end{align*}
\end{Definition}

%\begin{wrapfigure}{r}{0.37\textwidth}
%\vspace*{-0.2in}
%\centering
%\includegraphics[height=4.5cm]{./Figs/augment-yusu}
%\vspace*{-0.3in}
%\caption{White dots are tree nodes of $T_1^f$ and $T_2^g$. Green dots are newly augmented tree nodes in $\hatT_1^f$ and $\hatT_2^g$. 
%%Solid horizontal lines indicate levels passing through critical heights, while dashed ones are induced by critical height from the other tree. 
%\label{fig:augmentedtree}}
%\end{wrapfigure}
Now sort all levels in $\setSL_i$ in increasing order of their heights, denoted by $\setSL_1 = \{\slone{1}, \slone{2}, \ldots, \slone{m} \}$ and $\setSL_2 = \{\sltwo{1}, \ldots, \sltwo{m} \}$, respectively. The \emph{child-level} of \superlevel{} $\slone{i}$ (resp. $\sltwo{i}$) is $\slone{i-1}$ (resp. $\sltwo{i-1}$) for any $i\in [2, m]$; symmetrically, $\slone{i}$ (resp. $\sltwo{i}$) is the \emph{parent-level} of $\slone{i-1}$ (resp. $\sltwo{i-1}$). 
Let $h_1, \ldots, h_m$ be the sequence of height values for $\slone{1}, \slone{2}, \ldots, \slone{m}$; that is, $h_i = height(\slone{i})$. 
Similarly, let $\widehat{h}_1, \widehat{h}_2, \ldots, \widehat{h}_m$ be the corresponding sequence for $\sltwo{i}$'s.

Note that 
there is a one-to-one correspondence between \superlevel{}s in $\setSL_1$ and $\setSL_2$: specifically, for any $i\in [1, m]$, we have $\widehat{h}_i = h_i + \mydelta$. 
From now on, when we refer to the $i$-th \superlevel{}s of $\hatT_1^f$ and $\hatT_2^g$, we mean \superlevel{}s $\slone{i}$ and $\sltwo{i}$. 
Also observe that there is no tree node in between any two consecutive \superlevel{}s in either $T_1^f$ or in $T_2^g$ (all tree nodes are from some \superlevel{}s). See Figure \ref{fig:augmentedtree} (b) for an illustration. 

Next, we augment the tree $T_1^f$ (resp. $T_2^g$) to add points from all \superlevel{}s from $\setSL_1$ (resp. from $\setSL_2$) also as \emph{tree nodes}. 
The resulting \emph{augmented trees} are denoted by $\hatT_1^f$ and $\hatT_2^g$ respectively; obviously, $\hatT_1^f$ (resp. $\hatT_2^g$) has isomorphic underlying space as $T_1^f$ (resp. $T_2^g$), just with additional degree-2 tree nodes. 
In particular, $V(\hatT_1^f)$ (resp. $V(\hatT_2^g)$) is formed by all points from all \superlevel{}s in $\setSL_1$ (resp. $\setSL_2$). 
See Figure \ref{fig:augmentedtree} (b): In this figure, solid horizontal lines indicate levels passing through critical heights, while dashed ones are induced by critical height from the other tree.
In what follows, given a \superlevel{} $L$, we use $V(L)$ to denote the set of nodes from this level. 
Note that $V(\slone{m})$ and $V(\sltwo{m})$ each contain only one node, which is $root(\hatT_1^f)$ and $root(\hatT_2^g)$ respectively. 
Given a node $v$ from $\slone{i}$ (resp. $\sltwo{i}$), let $\levelC(v)$ denote its children nodes in the augmented tree. Each child node of $v$ must be from level $\slone{i-1}$ (resp. $\sltwo{i-1}$), as there are no tree-nodes between two consecutive \superlevel{}s. 

\begin{Definition}[Valid pair]
Given a node $w \in V(\hatT_2^g)$ and a collection of nodes $S \subseteq V(\hatT_1^f)$, we say that $(S, w)$ form a \emph{valid pair} if there exists an index $j\in [1, m]$ such that (1) $S \subseteq V(\slone{j})$ and $w \in V(\sltwo{j})$ (which implies that nodes in $S$ at height $h_j$ while $w$ has height $g(w) = \widehat{h}_j$); and (2) all nodes in $S$ have the same ancestor at height $h_j + 2\mydelta$ (which also equals $\widehat{h}_j + \mydelta$). Intuitively, it indicates that $S$ has the basic condition to be mapped to $w$ under some ${\eps}$-\mygood{} maps.

We say that $S$ is \emph{valid} if it participates some valid pair (and thus condition (2) above holds). 
%if condition (ii) above holds. 
%We also say that $S$ is \emph{valid} and $(S, w)$ is a \emph{valid-pair} from level-$j$. 
\end{Definition}

\paragraph{A first (slower) dynamic programming algorithm.}
We now describe our dynamic programming algorithm. 
To illustrate the main idea, we first describe a much cleaner but also slower dynamic programming algorithm \DPalg() below. Later in Section \ref{subsec:fasteralg} we modify this algorithm to improve its time complexity (which requires significant additional technical details). 

Our algorithm maintains a certain quantity, called \emph{feasibility} $F(S, w)$ for valid pairs in a bottom-up manner. 
Recall that we have defined the \mydepth{} of a node $u \in T^h$ in a merge tree $T^h$ as the height of the subtree $T^h(u)$ rooted at $u$; or equivalently
$\newdepth(u) = \max_{x \mydesceq u} |h(u) - h(x)|. $

\begin{description}\denselist
\item[Algorithm \DPalg($T_1^f, T_2^g, \mydelta$): ]
\item[Base case ($i=1$):] For each valid-pair $(S, w)$ from level-1, set $F(S, w) = 1$ (``true'') if and only if $\newdepth(w) \le 2\mydelta$; otherwise, set $F(S, w) = 0$ (``false''). 
\item[When $i > 1$:] Suppose we have already computed the feasibility values for all valid-pairs from level-($i-1$) or lower. 
%We now describe how to compute the feasibility value for all valid-pairs from level-$i$.  
Now for any valid-pair $(S, w)$ from level-$i$, 
we set $F(S,w) = 1$ if and only if the following holds: 
Consider the set of children $\levelC(S) \subseteq \slone{i-1}$ of nodes in $S$, and $w$'s children $\levelC(w) = \{w_1, \ldots, w_k\}$ in $\sltwo{i-1}$. 

If $\levelC(w)$ is empty, then $F(S,w)=1$ \emph{only if} $\levelC(S)$ is also empty; otherwise $F(S,w) = 0$. 

If $\levelC(w)$ is not empty, then we set $F(S,w)$=1 if there exists a partition of $\levelC(S) = S_1 \cup S_2 \cup \ldots \cup S_k$ (where $S_i \cap S_j = \emptyset$ for $i \neq j$, and it is possible that $S_i =\emptyset$) such that for each $j\in [1, k]$, 
\begin{itemize}\denselist
\item[(F-1)] if $S_j \neq \emptyset$, then $F(S_j, w_j) = 1$; and
\item[(F-2)]  if $S_j  = \emptyset$, then $\newdepth(w_j) \le 2\mydelta - (\widehat{h}_i- \widehat{h}_{i-1})$; note that this implies that $\widehat{h}_i- \widehat{h}_{i-1} \le 2\mydelta$ in this case.
\end{itemize}
\item[Output:] \DPalg($T_1^f, T_2^g, \mydelta$) returns ``yes'' if and only if $F(root(\hatT_1^f), root(\hatT_2^g)) = 1$. 
\end{description}
Recall that $root(\hatT_1^f)$ (resp. $root(\hatT_2^g)$) is the only node in $V(\slone{m})$ (resp. $V(\sltwo{m})$). 

We will first prove the following theorem for this slower.
%and \ref{subsec:timecomplexity}. 
In Section \ref{subsec:fasteralg} we show that time complexity can be reduced by almost a factor of $n$. 
\begin{theorem}\label{thm:DPalg}
(i) Algorithm \DPalg($T_1^f, T_2^g, \mydelta$) returns ``yes'' if and only if $d_I(T_1^f, T_2^g) \le \mydelta$. 

(ii) Algorithm \DPalg($T_1^f, T_2^g, \mydelta$) can be implemented to run in $O(n^3 2^{\myk}\myk^{\myk+1} )$ time, where $n$ is the total size of $T_1^f, T_2^g$, and $\myk = \myk_\mydelta(T_1^f, T_2^g)$ is the $\mydelta$-\degbound{} w.r.t. $T_1^f$ and $T_2^g$. 

Note that if $\myk$ is constant, then the time complexity is $O(n^3)$. 
\end{theorem}
%\begin{proof}
%The proof can be found in  Appendix %\ref{subsec:correctness}.
%\end{proof}

% \subsubsection{Proof of Theorem \ref{thm:DPalg}}
% \label{subsec:correctness}
%\paragraph{Proof of Theorem \ref{thm:DPalg}.} ~\\
In the remainder of this section, we sketch the proof of Theorem \ref{thm:DPalg}. 

\paragraph{Part (i) of Theorem \ref{thm:DPalg}: correctness.} 
We first show the correctness of algorithm \DPalg(). 
Give a subset of nodes $S'$ from some \superlevel{} of $\hatT_1^f$, let $\myF_1(S')$ denote the forest consisting of all subtrees rooted at nodes in $S'$. For a node $w' \in T_2^g$, let $T_2(w')$ denote the subtree of $T_2^g$ rooted at $w'$. 
We will now argue that $F(S, w) = 1$ if and only if there is a ``partial'' $\mydelta$-good map from $\myF_1(S) \to T_2(w)$. 

More precisely: 
a continuous map $\alpha: \myF_1(S) \to T_2(w)$ with $(S, w)$ being valid is a \emph{\mypartial{$\eps$}}{}map, if properties {\sf (P1)}, {\sf (P2)}, and {\sf (P3)} from Definition \ref{def:epsgood} hold (with $T_1^f$ replaced by $\myF_1(S)$ and $T_2^g$ replaced by $T_2(w)$). Note that in the case of {\sf (P2)}, the condition in {\sf (P2)} only needs to hold for $u_1, u_2 \in \myF_1(S)$ (implying that $\alpha(u_1), \alpha(u_2)\in T_2(w)$); that is, if $\alpha(u_1) \myanceq \alpha(u_2)$ for $u_1, u_2 \in \myF_1(S)$, then, we have $u_1^{2\eps} \myanceq u_2^{2\eps}$. 
Note that while $u_1$ and $u_2$ are from $\myF_1(S)$,  $u_1^{2\eps}$ and $u_2^{2\eps}$ may not be in $\myF_1(S)$ as it is possible that $f(u_1^{2\eps}) = f(u_1) + 2\eps \ge height(S)$.
First, we observe the following: 
\begin{claim}\label{claim:rootgood}
At the top level  where $\slone{m} = \{ u = root(\hatT_1^f) \}$ and $\sltwo{m} = \{w = root(\hatT_2^g) \}$ both contain only one node, if there is a \mypartial{$\mydelta$}{} map from $\myF_1(\{u\}) \to T_2(w)$, then there is a $\mydelta$-good map from $\uTone \to \uTwo$. 
\end{claim}
%\begin{proof}
%Let $\alpha: \myF_1(\{u\}) \to T_2(w)$ be the \mypartial{$\mydelta$}{} map. 
%Since $u = root(\hatT_1^f)$ and $w = root(\hatT_2^g)$, we can easily extend the map $\alpha$ to a map $\alpha': \uTone \to \uTwo$ by sending any point $x \myanc u$ from $\uTone$ to the unique ancestor $y \myanc w$ at height $g(y) = f(x) + \mydelta$. 
%It is easy to see that this map is $\mydelta$-good. 
%\end{proof} 

 The correctness of our dynamic programming algorithm (part (ii) of Theorem \ref{thm:DPalg}) will follow from Claim \ref{claim:rootgood} and Lemma \ref{lem:partialgood} below. Lemma \ref{lem:partialgood} is one of our key techincal results, and its proof can be found in Appendix \ref{appendix:lem:partialgood}. 
\begin{lemma}\label{lem:partialgood}
For any valid pair $(S, w)$, $F(S,w) =1$ if and only if there is a \mypartial{$\mydelta$}{} map $\alpha: \myF_1(S) \to T_2(w)$. 
\end{lemma}

%\subsection{Part (ii) of Theorem \ref{thm:DPalg}: time complexity}
%\label{subsec:timecomplexity}
\paragraph{Part (ii) of Theorem \ref{thm:DPalg}: time complexity.}
%\noindent\underline{Part (ii) of Theorem \ref{thm:DPalg}: time complexity.}
We now show that Algorithm \DPalg() can be implemented to run in the claimed time. 
Note that the augmented-tree nodes contain tree nodes of $T_1^f$ and $T_2^g$, as well as the intersection points between tree arcs of $T_1^f$ (resp. $T_2^g$) and \superlevel{}s. 
As there are at most $m = 2n$ number of \superlevel{}s in $\setSL_1$ and $\setSL_2$, it follows that the total number of tree nodes in the augmented trees $\hatT_1^f$ and $\hatT_2^g$ is bounded by $O(nm)=O(n^2)$. 
In what follows, in order to distinguish between the tree nodes for the augmented trees ($\hatT_1^f$ and $\hatT_2^g$) from the tree nodes of the original trees ($T_1^f$ and $T_2^f$), we refer to nodes of the former as \emph{augmented-tree nodes}, while the latter simply as \emph{tree nodes}. 
It is important to note that the $\mydelta$-\degbound{} is defined with respect to the original tree nodes in $T_1^f$ and $T_2^g$, not for the augmented trees (the one for the augmented trees can be significantly higher). 

Our DP-algorithm essentially checks for the feasibility $F(S, w)$ of valid-pairs $(S,w)$S. 
%To bound the size and number of valid pairs, we first state the following simple claim. 
The following two lemmas bound the size of valid pairs, and their numbers. Their proofs are in Appendix \ref{appendix:lem:sizevalids} and \ref{appendix:lem:coarsevalidbound}, respectively. 
%\noindent\emph{\underline{Remark 1:}~} If we implement \DPalg($T_1^f, T_2^g, \mydelta$) as described above, it will take $O(n^3 + n^2 2^{\myk}\myk^{\myk+1} )$ time. To achieve the stated time complexity we need to implement the case when $i > 1$ more carefully which we will describe in Section \ref{subsec:timecomplexity}. 

%\noindent\emph{\underline{Remark:}~}We remark that in fact, the time complexity is $\min\{ O(n^2 2^{\myk} \myk^{\myk+1}) , O(n^2 2^{n})\}$, and we omit the second term as it only becomes relevant when $\myk$ is very large (larger than $\Omega(n/ \log n)$). 

%In the remainder of this section, we prove part (i) and (ii) of Theorem \ref{thm:DPalg} in Sections \ref{subsec:correctness} and \ref{subsec:timecomplexity}, respectively. 

%\subsection{Proof of Theorem \ref{thm:DPalg}}
%\subsection{Part (i) of Theorem \ref{thm:DPalg}: Correctness}

%%%%%%%%%%%%%%%%%%%%
%%%%%%%%%%%%%%%%%%%%%%
%The following lemmas are proven in Appendix \ref{appendix:lem:sizevalids} and \ref{appendix:lem:coarsevalidbound} respectively. 

\begin{lemma}\label{lem:sizevalidS}
For any valid pair $(S, w)$, we have  $|S| \le \myk$ and $|\levelC(S)| \le \myk$, where $\myk = \myk_\mydelta(T_1^f, T_2^g)$ is the $\mydelta$-\degbound{} w.r.t. $T_1^f$ and $T_2^g$. 
\end{lemma}

%To bound the total number of valid-pairs, in what follows we first provide a coarser upper bound in Lemma \ref{lem:coarsevalidbound}. We then show how to modify Algorithm \DPalg() so as to reduce the number of valid pairs that the algorithm will inspect to $O(n^2 2^\myk)$, which eventually will lead us to the claimed time complexity in Theorem \ref{thm:DPalg}. 

%In what follows we will separate valid pairs to two classes: (i) a \emph{singleton-pair} $(S, w)$ is a valid pair with $|S|=1$, or (ii) a \emph{non-singleton-pair} is a valid pair $(S, w)$ with $|S| > 1$. 

\begin{lemma}\label{lem:coarsevalidbound}
Let $\myk = \myk_\mydelta(T_1^f, T_2^g)$ be the $\mydelta$-\degbound{} w.r.t. $T_1^f$ and $T_2^g$.
The total number of valid pairs that Algorithm \DPalg($T_1^f, T_2^g, \mydelta$) will inspect is bounded by $O(n^3 2^\myk)$, and they can be computed in the same time. 
%\\
%(i) The number of non-singleton-pairs $(S, w)$ is bounded by $O(n^2 2^\myk)$. \\
%(ii) The number of singleton-pairs $(S, w)$ is bounded by $O(n^3)$.
\end{lemma}

To obtain the final time complexity for Algorithm \DPalg{}, 
consider computing $F(S,w)$ for a fixed valid pair $(S,w)$. This takes $O(1)$ time in the base case (the \superlevel{} index $i=1$). Otherwise for the case $i>1$, observe that $k=|\levelC(w)| = 
degree(w) \le \myk$, and $|\levelC(S)| \le \myk$ by Lemma \ref{lem:sizevalidS}. Hence  the number of partitioning of $\levelC(S)$ is bounded by 
$O(|\levelC(S)|^{k}) = O(\myk^\myk)$. 
For each partition, checking conditions (F-1) and (F-2) takes $O(k)$ time; thus the total time needed to compute $F(S,w)$ is $O(k \myk^\myk) = O(\myk^{\myk+1})$. 
Combining this with Lemma \ref{lem:coarsevalidbound}, we have that the time complexity of Algorithm \DPalg() is bounded from above by 
$O({n^3}{2^{\myk}}\myk^{\myk+1})$, as claimed. 
%\paragraph{Reducing the number of valid pairs inspected.} 

\subsection{A Faster Algorithm}
\label{subsec:fasteralg}

It turns out that we do not need to inspect all the $O(n^3 2^\myk)$ number of valid pairs as claimed in Lemma \ref{lem:coarsevalidbound}. 
We can consider only what we call \sensiblepair{}s, which we define now. 

\begin{Definition}\label{def:sensiblepair} 
Given a valid-pair ($S, w$), suppose $S$ is from \superlevel $\slone{i}$ and thus $w$ is from \superlevel{} $\sltwo{i}$. Then, $(S, w)$ is a \emph{sensiblepair}{} if either of the following two conditions hold: 
\begin{itemize}\denselist
    \item[(C-1)] $S$ contains a tree node from $V(T_1^f)$, or its children $\levelC(S) \subseteq \slone{i-1}$ in the augmented tree $\hatT_1^f$ contains some tree node from $V(T_1^f)$, or the parents of nodes of $S$ in the augmented tree $\hatT_1^f$ (which are necessarily from \superlevel{} $\slone{i+1}$) contains some tree node from $V(T_1^f)$; or
    \item[(C-2)] $w$ is a tree node of $T_2^g$, or $\levelC(w) \subseteq \sltwo{i-1}$ contains a tree node of $T_2^g$; or the parent of $w$ from \superlevel{} $\sltwo{i+1}$ in the augmented tree $\hatT_2^g$ is a tree node of $T_2^g$. 
%    \item[(iii)] the parents of $S$ in the augmented tree $\hatT_1^f$ (i.e, the parents of $S$ from \superlevel{} $\slone{i+1}$) contains a tree node from $T_1^f$; or
%    \item[(iv)] the parent of $w$ from \superlevel{} $\sltwo{i+1}$ in the augmented tree $\hatT_2^g$ is a tree node from $T_2^g$. 
\end{itemize}

%We say that $S$ is a \emph{sensible-set} if it satisfies condition (i) above. 
%if it participates in a \sensiblepair{}. 
\end{Definition}

Algorithm \DPalg() can be modified to Algorithm \modifyDPalg() so that it only inspects \sensiblepair{}s. The modification is non-trivial, and the reduction in the bound on number of \sensiblepair{}s is by relating \sensiblepair{}s to certain appropriately defined edge-list pairs $(A\subseteq E(T_1^f), \alpha \in E(T_2^g))$. 
The rather technical details can be found in Appendix \ref{appendix:subsec:fasteralg}. 
We only summarize the main theorem below. 
\begin{theorem}\label{thm:fasteralg}
(i) Algorithm \modifyDPalg($T_1^f, T_2^g, \mydelta$) returns ``yes'' if and only if $d_I(T_1^f, T_2^g) \le \mydelta$. \\
(ii) Algorithm \modifyDPalg($T_1^f, T_2^g, \mydelta$) can be implemented to run in $O(n^2 2^\myk \myk^{\myk+2}\log n)$ time, where $n$ is the total complexity of input trees $T_1^f$ and $T_2^g$, and $\myk = \myk_\mydelta(T_1^f, T_2^g)$ is the $\mydelta$-\degbound{} w.r.t. $T_1^f$ and $T_2^g$. 

Note that if $\myk$ is constant, then the time complexity is $O(n^2\log n)$.
\end{theorem}

\section{Algorithms for Interleaving and Gromov-Hausdorff Distances}\label{sec:optimization}

\subsection{FPT Algorithm to Compute Interleaving Distance}
\label{subsec:int-optimizaton}

In the previous section, we show how to solve the decision problem for interleaving distance between two merge trees $T^f_1$ and $T_2^g$. 
We now show how to compute the interleaving distance $\Intdopt$, which is the smallest $\mydelta$ value such that $d_I(T_1^f, T_2^g) \le \mydelta$ holds. 

The main observation is that there exists a set $\criSet$ of $O(n^2)$ number of \emph{candidate values} such that $\Intdopt$ is necessarily one of them. 
Specifically, let $\criSet_1 =  \{ |f(u) - g(w)| \mid u \in V(T_1^f), w \in V(T_2^g) \}$, 
$\criSet_2 = \{ |f(u) - f(u')|/2 \mid u, u' \in V(T_1^f) \}$, and 
$\criSet_3 = \{ |g(w) - g(w') | / 2 \mid w, w' \in V(T_2^g) \}$. 
Set $\criSet = \criSet_1 \cup \criSet_2 \cup \criSet_3$. 
The proof of the following lemma can be found in Appendix \ref{appendix:lem:criSet}. 
\begin{lemma}\label{lem:criSet}
The interleaving distance $\Intdopt = d_I(T_1^f, T_2^g)$ satisfies that $\Intdopt \in \criSet$. 
\end{lemma}

Finally, compute and sort all candidate values in $\criSet$ where by construction, $|\criSet| = O(n^2)$. Then, starting with $\delta$ being the smallest candidate value in $\criSet$, we perform algorithm \DPalg($T_1^f, T_2^g, \delta)$ for each $\delta$ in $\criSet$ in increasing order, till the first time the answer is `yes'. The corresponding $\delta$ value at the time is $d_I(T_1^f, T_2^g)$. Furthermore, note that for the degree-bound parameter, $\tau_{\delta}(T_1^f, T_2^g) \le \tau_{\delta'} (T_1^f, T_2^g)$ for $\delta \le \delta'$. Combining with Theorem \ref{thm:fasteralg}, we can easily obtain the following trivial bound: 
\begin{theorem}\label{thm:optinterleaving-slow}
Let $\Intdopt = d_I(T_1^f, T_2^g)$ and $\tau^* = \tau_\Intdopt(T_1^f, T_2^g)$ be the degree-bound parameter of $T_1^f$ and $T_2^g$ w.r.t. $\Intdopt$. 
Then we can compute $\Intdopt$ in $O(n^4 2^{\tau^*} (\tau^*)^{\tau^*+2}\log n)$ time. 
\end{theorem}

However, it turns out that one can remove almost an $O(n^2)$ factor by using a double-binary search like procedure, as discovered by Kyle Fox. We include this improved result below and his argument below for completeness. See Appendix \ref{appendix:thm:optinterleaving-fast} for the proof. 
\begin{theorem}\label{thm:optinterleaving-fast}
Let $\Intdopt = d_I(T_1^f, T_2^g)$ and $\tau^* = \tau_\Intdopt(T_1^f, T_2^g)$ be the degree-bound parameter of $T_1^f$ and $T_2^g$ w.r.t. $\Intdopt$. 
Then we can compute $\Intdopt$ in $O(n^2 2^{2\tau^*} (2\tau^*)^{2\tau^*+2}\log^3 n)$ time. 
\end{theorem}

%\noindent {\it \underline{Remarks}:} We note that we cannot directly use binary search to identify the optimal $\Intdopt$ from $\criSet$, because if we perform \DPalg($T_1^f, T_2^g, \delta$) for a $\delta$ value larger than $\Intdopt$, its degree-bound parameter could be much larger, making our algorithm not fixed-parameter tractable any more. It will be an interesting open problem to see how this quadratic blow-up can be reduced. 

%\paragraph{FPT-algorithm for Gromov-Hausdorff distance.} 
\subsection{FPT-Algorithm for Gromov-Hausdorff Distance}

Finally, we develop a FPT-algorithm to approximate the Gromov-Hausdorff distance between two input trees $(T_1, d_1)$ and $(T_2, d_2)$. 
%As there is no function defined on the trees any more, we need to modify our parameter slightly. 
To approximate the Gromov-Hausdorff distance between two metric trees, we need to modify our parameter slightly (as there is no function defined on input trees any more). 
Specifically, now given a metric tree $(T, d)$, a $\eps$-geodesic ball at $u\in |T|$ is simply $\widehat{B}_\eps(u, T) = \{ x\in |T| \mid d(x, u) \le \eps \}$. 
\begin{description}
\item[Parameter $\tau$:] Given $\Tcal_1 = (T_1, d_1)$ and $\Tcal_2 = (T_2, d_2)$, define the \emph{$\eps$-metric-degree-bound} parameter $\newtau_\eps(T_1, T_2)$ to be the largest sum of degrees of all tree nodes within any $\eps$-geodesic ball in $T_1$ (w.r.t. metric $d_1$) or in $T_2$ (w.r.t. $d_2$). 
\end{description}

We obtain our main result for approximating the Gromov-Hausdorff distance between two metric trees within a factor of $14$. 
We note that to obtain this result, we need to also relate the $\eps$-metric-degree-bound parameter for metric trees with the $\eps$-degree-bound parameter used for interleaving distance for the special geodesic functions we use (in fact, we will show that $\newtau_\delta \le \tau_\delta \le \newtau_{2\delta}$). The proof of the following main theorem of this section can be found in Appendix \ref{appendix:thm:GHalg}. 
\begin{theorem}\label{thm:GHalg}
Given two metric trees $\Tcal_1 = (T_1, d_1)$ and $\Tcal_2 = (T_2, d_2)$ where the total number of vertices of $T_1$ and $T_2$ is $n$, 
we can $14$-approximate the Gromov-Hausdorff distance $\hat{\Intdopt} = \dgh(\Tcal_1, \Tcal_2)$ in $O(n^4\log n + n^2 2^\newtau \newtau^{\newtau+2}\log^3 n)$ time, where $\newtau = 2\newtau_{28\hat{\Intdopt}}(T_1, T_2)$ is twice the metric-degree-bound parameter w.r.t. $28\hat{\Intdopt}$. 
%
%If $\newtau$ is a constant, then the time complexity is $O(n^6)$. 
\end{theorem}

\paragraph{Remarks:} 
We remark that the time complexity of the FPT approximation algorithm of \cite{Sch17} contains terms $n^k$, where $k$ is the parameter and could be large in general -- Indeed, $k$ is the cardinality of an $\varepsilon$-net of one of the input metric spaces, and $\varepsilon$ also appears as an additive approximation term for  algorithm. In contrast, the dependency of our algorithm on the parameter $\hat{\tau}$ is roughly $O(2^{O(\hat{\tau})})$, and our algorithm has only constant multiplicative approximation factor. On the other hand, note that the algorithm of \cite{Sch17} works for general finite metric spaces. 
%First, we remark that it is maybe possible to reduce the $O(n^4\log n)$ term in the above time complexity by allowing further approximation of the Gromov-Hausdorff distance $\hat{\Intdopt}$. 
We also remark that the Gromov-Hausdorff distance between two metric spaces $(X, d_X)$ and $(Y, d_Y)$ measures their \emph{additive distortion}, and thus is not invariant under scaling. In particular, suppose the input two metric spaces $\Tcal_1 = (T_1, d_1)$, $\Tcal_2 = (T_2, d_2)$ scale by the same amount to a new pair of input trees $\Tcal_1' = (T_1', d_1' = c\cdot d_1)$, $\Tcal_2' = (T_2', d_2' = c \cdot d_2)$. Then the new Gromove-Hausdorff distance between them $\delta_{GH} ( \Tcal_1' , \Tcal_2') = c \cdot \delta_{GH} (\Tcal_1, \Tcal_2)$. 
However, note that the metric-degree-bound parameter for the new trees satisfies $\newtau_{c\delta}(\Tcal_1', \Tcal_2') = \newtau_\delta(\Tcal_1, \Tcal_2)$. Hence the time complexity of our algorithm to approximate the Gromov-Hausdorff distance $\delta_{GH} ( \Tcal_1' , \Tcal_2')$ for scaled metric-trees $\Tcal_1'$ and $\Tcal_2'$ {\bf remains the same} as that for approximating the Gromov-Hausdorff distance $\delta_{GH} ( \Tcal_1, \Tcal_2)$.

\section{Concluding Remarks}
\label{sec:conclusion}

In this paper, by re-formulating the interleaving distance, we developed the first FPT algorithm to compute the interleaving distance \emph{exactly} for two merge trees, which in turn leads to an FPT algorithm to approximate the Gromov-Hausdorff distance between two metric trees. 

We remark that the connection between the Gromov-Hausdorff distance and the interleaving distance is essential, as the interleaving distance has more structure behind it, as well as certain ``order'' (along the function associated to the merge tree), which helps to develop dynamic-programming type of approach. 
For more general metric graphs (which represent much more general metric spaces than trees), it would be interesting to see whether there is a similar relation between the Gromov-Hausdorff distance of metric graphs and the interleaving distance between the so-called Reeb graphs (generalization of merge trees). 
%Instead of equivalence up to a constant (Claim \ref{claim:GHandInterleaving}), it is possible that they approximate each other up to a factor depending on the hyperbolicity as well as the genus of the graph. 
%Developing efficient algorithms to compute the Gromov-Hausdorff distance between other non-trivial metric spaces also remain open and very interesting. 

%\newpage 
\paragraph{Acknowledgment:} We thank reviewers for helpful comments. We would like to thank Kyle Fox, for suggesting an elegant double-binary search procedure, to improve the time complexity of the optimal interleaving distance by almost a factor of $n^2$ (see Theorem \ref{thm:optinterleaving-fast}), which further leads to a similar improvement for approximating the Gromov-Hausdorff distance (Theorem \ref{thm:GHalg}). 
This work is partially supported by National Science Foundation (NSF) under grants CCF-1740761, IIS-1815697 and CCF-1618247, as well as by National Institute of Health (NIH) under grant R01EB022899.   

\bibliographystyle{abbrv}
\bibliography{ref}

\newpage
\appendix

\section{Proof of Theorem \ref{thm:epsgoodtwo}}
\label{appendix:thm:epsgoodtwo}

Theorem \ref{thm:epsgoodtwo} follows from Lemma \ref{lem:epsgoodone} and \ref{lem:epsgoodtwo} below. 
%Their proofs can be found in Appendix \ref{appendix:lem:epsgoodone} and \ref{appendix:betacontinuous} respectively. 
\begin{lemma}\label{lem:epsgoodone}
If $d_I(T_1^f, T_2^g) \le \eps$, then there exists an ${\eps}$-\mygood{} map ${\alpha}:\uTone \rightarrow \uTwo$.
\end{lemma}
\begin{proof}
By definition, if $d_I(T_1^f, T_2^g) \le \eps$, then there exists a pair of $\eps$-compatible maps $\alpha_\eps: \uTone \to \uTwo$ and $\beta_\eps: \uTwo \to \uTone$. 
We simply set $\alpha := \alpha_\eps$, and argue that the three properties in Definition \ref{def:epsgood} will all hold for $\alpha$. 

Indeed, property {\sf (P1)} follows trivially from condition (C1) of Definition \ref{def:compatible}.
To see that property {\sf (P2)} holds, set $w_i = \alpha(u_i)$, for $i = 1, 2$, and note that $u_i^{2\eps} =\beta_\eps(w_i)$. 
Since the map $\beta_\eps$ is continuous and $f(\beta_\eps(w)) = g(w) + \eps$ for any $w\in \uTwo$ (therefore monotone), it then follows from Observation \ref{obs:ancester} that $u_1^{2\eps} = \beta_\eps(w_1) \myanceq \beta_\eps(w_2) = u_2^{2\eps}$, establishing property {\sf (P2)}. 
 
We now show that property {\sf (P3)} also holds. 
Specifically, consider any $w \in \uTwo \setminus \Img(\alpha)$, and let $\mywF$ be its lowest ancestor from $\Img(\alpha)$. 
%Since $w^F \in \Img(\alpha)$, there exists $u \in T_1^f$ such that $w^F = \alpha(u)$. 
%Since $\alpha (= \alpha_\eps)$ and $\beta_\eps$ are $\eps$-compatible, we have that $\beta_\eps(w^F) = u^{2\eps}$. 
Assume on the contrary that $g(\mywF) - g(w) > 2 \eps$, then it must be that $\mywF \myanc w^{2\eps} \myanc w$. 
On the other hand, consider $u  = \beta_\eps(w)$. As $\alpha (= \alpha_\eps)$ and $\beta_\eps$ are $\eps$-compatible, we have that $\alpha(u) = w^{2\eps}$, meaning that $w^{2\eps} \in \Img(\alpha)$. 
This however contradicts our assumption that $\mywF$ is the lowest ancestor of $w$ from $\Img(\alpha)$ as $\mywF \myanc w^{2\eps}$. Hence it is not possible that $g(w^F) - g(w) > 2\eps$, and property {\sf (P3)} holds. 
\end{proof}

\begin{lemma}\label{lem:epsgoodtwo}
If there is an $\eps$-\mygood{} map $\alpha: \uTone \to \uTwo$, then $d_I(T_1^f, T_2^g) \le \eps$. 
\end{lemma}

\begin{proof}
%%%%%%%%%%%%%%%%%%%%%%%%%
%%%%%%%%%%%%%%%%%%%%%%%
%\noindent
Given an $\eps$-\mygood{} map $\alpha: \uTone \to \uTwo$, we will show that we can construct a pair of continuous maps $\alpha_\eps: \uTone \to \uTwo$ and $\beta_\eps: \uTwo \to \uTone$ that are $\eps$-compatible, which will then prove the lemma. 
Specifically, set $\alpha_\eps = \alpha$. 
We construct $\beta_\eps$ as follows: 

\begin{figure}[tpbh]
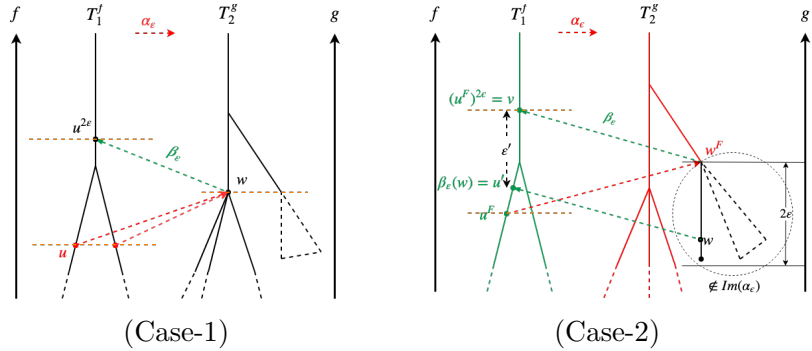

\begin{center}
\begin{tabular}{cc}
\includegraphics[height=4cm]{./Figs/lemma2case3} & \includegraphics[height=4cm]{./Figs/lemma2case2}\\
(Case-1) & (Case-2)
\end{tabular}
\vspace*{-0.2in}
\end{center}
\caption{(Case-1). Construction of $\beta_{\epsilon}$ for the points of $T_2^g$ which are in $\Img(\alpha)$. 
(Case-2). Construction of $\beta_{\epsilon}$ for the points of $T_1^f$ which are not in $\Img(\alpha)$. $w$ is not in the $\Img(\alpha)$, and $w^F$ is its lowest ancestor in the $\Img(\alpha)$. $v=\beta_\varepsilon(w^F)$ is a point in $2\varepsilon$ higher than $u^F$ in $T_1^f$. $u'= \beta_{\varepsilon}(w)$ is an ancestor of $u^F$ that we map $w$ by $\beta_{\varepsilon}$.
\label{fig:Lemma2}}
\end{figure}
\noindent\emph{(Case-1):} For any point $w\in \Img(\alpha) (\subseteq \uTwo)$, we simply set $\beta_\eps(w) = u^{2\eps}$, where $u \in \uTone$ is any point from $\alpha^{-1}(w)$. 
Note that the choice of $u$ does not matter. This is because that as $\alpha(u_1) = \alpha(u_2) =w$ (thus $\alpha(u_1) \myanceq \alpha(u_2))$, it then follows from 
property {\sf (P2)} of the $\eps$-\mygood{} map $\alpha$ that $u_1^{2\eps} \myanceq u_2^{2\eps}$. Since $f(u_1^{2\eps}) = f(u_2^{2\eps}) = f(u_1)+2\eps$, it then must be that $u_1^{2\eps} = u_2^{2\eps}$. Hence $\beta_\eps(w)$ is well-defined (independent of the choice of $u$). See the right figure in Figure \ref{fig:Lemma2}.  

\noindent\emph{(Case-2):} For any point $w\in \uTwo \setminus \Img(\alpha)$ (i.e, $w\notin \Img(\alpha)$), let $\mywF$ be its lowest ancestor in $\Img(\alpha)$. 
There could be multiple points in $T_1^f$ from $\alpha^{-1}(\mywF)$. Consider an arbitrary but fixed choice $u^F\in \alpha^{-1}(\mywF)$: For example, assume that all nodes and edges from $T_1^f$ have a unique integer id, and we choose $u^F$ to be the point from the node or edge with lowest id. 
Set $v = (u^F)^{2\eps}$ to be the ancestor of $u^F$ at height $2\eps$ above $u^F$. From (Case-1), we know that we have already set $\beta_\eps(\mywF) = v$. 
See the above figure for an illustration.

Let $\eps' = g(\mywF) - g(w)$; by property {\sf (P3)}, $\eps' \le 2\eps$. 
We simply set $\beta_\eps(w)$ to be the unique point $u'$  from the path connecting $u^F$ to its ancestor $v$ such that $f(v) - f(u') = \eps'$; that is, $v \myanceq u'  \myanceq u^F$. 
Easy to verify that under this construction, $f(\beta_\eps(w)) = g(w) + \eps$. 

We now show that $\alpha_\eps$ and $\beta_\eps$ as constructed above form a pair of $\eps$-compatible maps for $T_1^f$ and $T_2^g$. 
First, we claim that $\beta_\eps$ is continuous 
(Note that $\alpha_\eps$ is continuous as $\alpha_\eps = \alpha$.) Next we show that all conditions in Definition \ref{def:compatible} are satisfied. 
%%%%%%%%%%%

\vspace*{0.15in}\emph{(Claim-1)
The map $\beta_\eps: \uTwo \to \uTone$ is continuous.}
To this end, we first put a function-induced metric $d_f$ (resp. $d_g$) on $T_1^f$ (resp. on $T_2^g$), defined as follows: 
for any $u, u' \in \uTone$, let $\pi(u, u')$ be the unique tree path connecting $u$ to $u'$. Set $d_f(u, u') = \max_{x\in \pi(u, u')} f(x) - \min_{y\in \pi(u,u')} f(y)$.
In other words, $d_f(u, u')$ is the maximum variation of the $f$-function value along the unique path $\pi(u, u')$ from $u$ to $u'$. 
Define the function-induced metric $d_g$ for $\uTwo$ in a symmetric manner. 

Let $B_r(x, T_1^f) = \{ y\in \uTone \mid d_f(x, y) < r\}$ denote the open ball around a point $x\in T_1^f$ under the function-induced metric $d_f$; 
% see Figure \ref{fig:Geodesic_Ball}, 
and define $B_r(z, T_2^g)$ symmetrically. 
(Note that this ball is in fact the same as the $\eps$-ball we will use at the beginning of Section \ref{sec:decision} to introduce the degree-bound parameter.) 
For simplicity, we sometimes omit the reference to the tree in these open balls when its choice is clear. 
To show that $\beta_\eps: \uTwo \to \uTone$ is continuous, we just need to show that for any $w \in \uTwo$ and $u = \beta_\eps(w)$, 
given any radius $r>0$, there always exists $r'>0$ such that $\beta_\eps(B_{r'} (w, T_2^g)) \subseteq B_r(u, T_1^f)$.

Fix $w\in \uTwo$, $u = \beta_\eps(w)$, and radius $r>0$. Let $0< r'  < r$ be a sufficiently small value so that $B_{r'}(w, T_2^g)$ contains no tree nodes other than potentially $w$. 
%See Figure \ref{fig:continuousbeta} for an illustration. 
We will prove that $\beta_\eps(w_0) \in B_{r'} (u, T_1^f)$ (and thus in $B_r(u, T_1^f)$) for any $w_0 \in \beta_\eps(B_{r'} (w, T_2^g))$. 
Note first, by our choice of $r'$, the unique path $\pi$ connecting $w_0$ to $w$ is monotone in $g$-function values, as the ball $B_{r'}(w, T_2^g)$ does not contain any tree node other than potentially $w$. 
%\yusu{Add a figure here.}
If $\pi \subset \Img(\alpha)$, then, by our construction of $\beta_\eps$, $\beta_\eps(\pi)$ is a monotone path and thus lies in $B_{r'} (u, T_1^f)$. Hence $\beta_\eps(w_0) \in B_{r'} (u, T_1^f)$ for any $w_0 \in B_{r'} (w, T_2^g)$. 

Now assume that $\pi \setminus \Img(\alpha) \neq \emptyset$. 
W.o.l.g assume that $g(w_0) < g(w)$; otherwise, we simply switch the role of the two in our argument below. 
Let $w^F$ denote the lowest ancestor of $w_0$ from $\Img(\alpha)$.
First, assume that $w \notin \Img(\alpha)$, which means that the entire path $\pi \in \uTwo \setminus \Img(\alpha)$, and $w^F \myanceq w \myanceq w_0$. 
%By construction of $\beta_\eps$, we note that the entire path $\pi$ between $w_0$ to $w$ is mapped into a canonical path $\pi' (u^F, v)$ where $\alpha(u^F) = w^F$ and $v = (u^F)^{2\eps} = \beta_\eps(w^F)$. 
In this case, let  
%$\beta_\eps(\pi) \subseteq \pi'(u^F, v)$ where 
$u^F \in \alpha^{-1}(w^F)$ be the preimage of $w^F$ under $\alpha$ used in our procedure to construct $\beta_\eps(w)$. Set $v = (u^F)^{2\eps} = \beta_\eps(w^F)$, and let $\pi'(u^F, v)$ denote the unique path between them (note that the path $\pi'(u^F, v)$ is monotone.) 
Based on our procedure to construct $\beta_\eps$, we know that both $\beta_\eps(w)$ and $\beta_\eps(w_0)$ are in $\pi'(u^F, v)$, and $\beta_\eps(\pi)$ is in fact a subpath of $\pi'(u^F, v)$ and thus also monotone. 
%and $u^F \mydesceq~ \beta_\eps(w_0) \mydesceq~ \beta_\eps(w) \mydesceq~  v = \beta_\eps(w^F)$. 
It then follows that  $\beta_\eps(\pi) \subseteq B_{r'} (u, T_1^f)$. 

The only remaining case is when $w \myanceq w^F \myanceq w_0$. In this case, we consider the two sub-path $\pi_1 = \pi(w_0, w^F)$ and $\pi_2 = \pi(w^F, w)$ of $\pi$. Using similar arguments above, we can show that $\beta_\eps(\pi)$ is still a monotone continuous path in $\uTone$ and thus $\beta_\eps(\pi) \in B_{r'} (u, T_1^f)$. Hence $\beta_\eps(w_0) \in B_{r'} (u, T_1^f)$ for any $w_0 \in B_{r'} (w, T_2^g)$.

Putting everything together, it then follows that the map $\beta_\eps$ constructed is continuous. 

%%%%%%%%%%%%%
%%%%%%%%%%%%%
%%%%%%%%%%%%%

%%%%%%%%%%%
%%%%%%%%%%%
\vspace*{0.15in}\emph{(Claim-2)
All conditions in Definition \ref{def:compatible} are satisfied for the map $\beta_\eps: \uTwo \to \uTone$.}
Conditions ({\sf C1}), ({\sf C2}) and ({\sf C3}) follow from the constructions of $\alpha_\eps$ and $\beta_\eps$. 
What remains is to prove that condition ({\sf C4}) holds. 
Consider any $w\in \uTwo$. 

If $w = \alpha(u) \in \Img(\alpha)$, then by construction $\beta_\eps(w) = u^{2\eps} \myanc u$. Then Observation \ref{obs:ancester} implies that $\alpha_\eps(\beta_\eps(w))$ is necessarily an ancestor of $\alpha(u) = w$. Furthermore, since $g(\alpha_\eps(\beta_\eps(w)) - g(w) = 2\eps$ (by condition ({\sf C3})), it then must be that $\alpha_\eps(\beta_\eps(w)) = w^{2\eps}$, establishing condition ({\sf C4}). 

Otherwise, $w\notin \Img(\alpha)$: let $\mywF$ be its lowest ancestor from $\Img(\alpha)$ with $u^F$ being the point from $\alpha^{-1}(\mywF)$ as used in (Case-2) of the construction above. Recall that we set $\beta_\eps(w)$ such that $(u^F)^{2\eps} \myanceq \beta_\eps(w) \myanceq u^F$. Hence $\alpha_\eps(\beta_\eps(w))$ must be an ancestor of $\alpha(u^F) = \mywF$. It then follows that $\alpha_\eps(\beta_\eps(w)) \myanceq \mywF \myanceq w$. Furthermore, since $g(\alpha_\eps \circ \beta_\eps(w)) = g(w) + 2\eps$, we have that $\alpha_\eps\circ \beta_\eps(w) = w^{2\eps}$. 
Putting everything together, we thus have that $\alpha_\eps$ and $\beta_\eps$ form a pair of $\eps$-compatible maps for $T_1^f$ and $T_2^g$, implying that $d_I(T_1^f, T_2^g) \le \eps$. 

This finishes the proof of Lemma \ref{lem:epsgoodtwo}. 
\end{proof}

\section{Proof of Lemma \ref{lem:partialgood}}
\label{appendix:lem:partialgood}

Recall that $h_i = height(\slone{i})$ and $\widehat{h}_i = height(\sltwo{i}) = h_i + \mydelta$, for any $i \in [1, m]$. We prove this lemma by induction on the indices of \superlevel{}s. 

For the base case when $i=1$, it is easy to verify that the lemma holds for any valid pair $(S,w)$ from level-1. 
Indeed, at this first level, $\myF(S) = S$ and $T_2(w) = w$, and thus nodes in $S$ and $w$ are all leaf-nodes of $T_1^f$ or $T_2^g$. Hence if $F(S,w) = 1$, properties {\sf (P1)} and {\sf (P3)} hold trivially for the map $\alpha: \myF(S) \to T_2(w)$ defined as $\alpha(s) = w$ for any $s\in \myF(s) =S$. Property {\sf (P2)} follows from the fact that $(S,w)$ is valid, thus $s_1^{2\mydelta} = s_2^{2\mydelta}$ for any two $s_1, s_2\in S$.  
Similarly, if there is a \mypartial{$\mydelta$}{} map $\alpha: \myF(S) \to T_2(w) = w$ for a valid pair $(S,w)$, it has to be that $\alpha(s) = w$ for any $s \in \myF(S)$. 
Since $T_2(w) = w$, $w$ is a leaf and thus $\newdepth(w) = 0 < 2\mydelta$, meaning that $F(S,w) = 1$. 

Now consider a generic level $i>1$, and assume the claim holds for all valid pairs from level $j < i$. We prove that the claim holds for any valid pair $(S, w)$ from level-$i$ as well: 

\begin{figure}[tpbh]
\begin{center}
\begin{tabular}{cc}
\includegraphics[height=3.8cm]{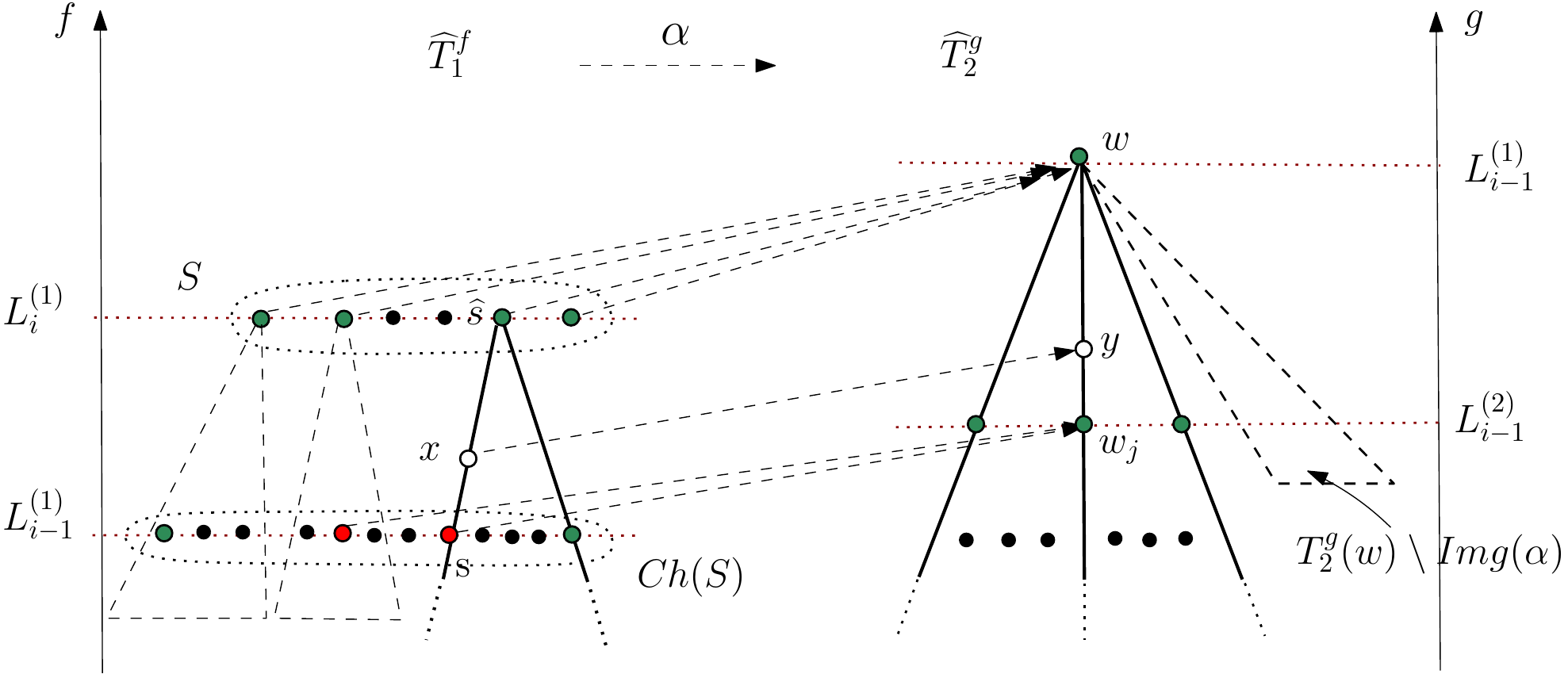} &
\includegraphics[height=3.7cm]{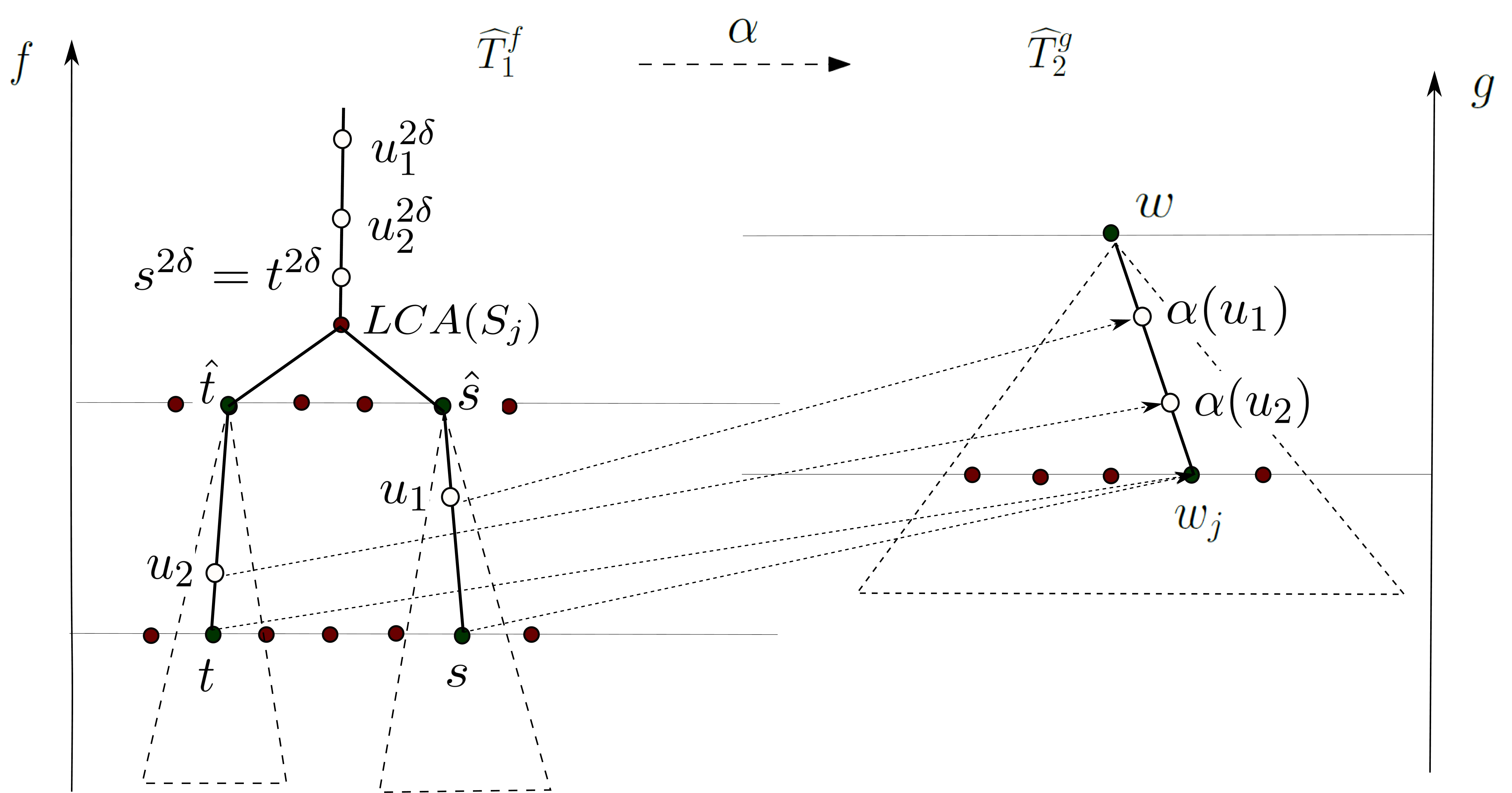}\\
(a) & (b)
\end{tabular}
\end{center}
\vspace*{-0.15in}
\caption{(a). 
Assume that $S_j \subset \levelC(S)$ consists the two red points $s$ and $s'$; note $\alpha(S_j) = w_j$. Under extension of $\alpha$, $x$ (from edge $(s, \hat{s})$) is mapped to $y \in (w_j, w)$. 
(b). Illustration for case-(a): $s, t \in S_j$ and $\alpha(S_j) = w_j$, hence $S_j$ must be valid, meaning that $s^{2\delta} = t^{2\delta} \myanceq LCA(S_j)$, which ulitmately implies that $u_1^{2\delta} \myanceq u_2^{2\delta}$.  
\label{fig:Lemma3-1}}
\end{figure} 
\begin{description}\denselist
\item[$\Rightarrow$: ] Suppose $F(S,w) = 1$. In this case, we will show that we can construct a \mypartial{$\mydelta$}{} map $\alpha: \myF_1(S) \to T_2(w)$. 
We assume that $w$ is not a leaf-node; as otherwise, $F(S,w)$ means that all nodes in $S$ are also necessarily leaf-nodes for $\hatT_1^f$ and thus $\myF(S) = S$ and $T_2(w) = \{w\}$. We then simply set $\alpha: \myF(S) \to T_2(w)$ as $\alpha(s) = w$ for each $s\in S$ and easy to see that this $\alpha$ is \mypartial{$\mydelta$}{} (by using the same argument as for the base case). 
Now suppose $\levelC(w) = \{w_1, \ldots, w_k\}$ and let $S_1, \ldots, S_k$ be the partition of $\levelC(S)$ that make $F(S,w) = 1$. 
Assume w.o.l.g that $S_{1}, \ldots, S_a$ are non-empty, while $S_{a+1}, \ldots, S_k$ are empty. 
Since $F(S_j, w_j) = 1$ for $j\in [1, a]$, there exists an \mypartial{$\mydelta$}{} map $\alpha_j: \myF_1(S_j) \to T_2(w_j)$ by induction hypothesis. The restriction of $\alpha$ to each $\myF_1(S_j)$ is simply $\alpha_j$. 
Then, we ``extend'' these $\alpha_j$'s into a map $\alpha: \myF_1(S) \to T_2(w)$ as follows: 

For each child $s \in \levelC(S)$, suppose $s\in S_j$ and the parent of $s$ is $\hat{s} \in S$. Under the \mypartial{$\mydelta$}{} map $\alpha_j: \myF_1(S_j) \to T_2(w_j)$ (from induction hypothesis), we know that $\alpha_j(s) = w_j$. 
We extend $\alpha_j$ to all points within segment $(s, \hat{s})$. Specifically, for any point $x\in edge(s, \hat{s})$, we simply set $\alpha(x)$ to be the corresponding point $y \in edge(w_j, w)$ at height $f(x) + \mydelta$ (i.e, $g(y)  = f(x) + \mydelta$). See Figure \ref{fig:Lemma3-1} (a) for an illustration.

Easy to verify that this extended map is continuous: as first, the extension along each edge $(s, \hat{s})$ is continuous. 
The only place where discontinuity may happen is at points $\hat{s}$ from $S$. However, all points in $S$ will be mapped to $w$ under this extension. Hence $\alpha$ constructed above is continuous. 

Furthermore, by construction, $\alpha$ satisfies property {\sf (P1)}. 
To prove that property {\sf (P3)} holds, consider any point $z \in T_2(w) \setminus \Img(\alpha)$.  
First, observe that by construction, all edge segments $[w_\ell, w]$, for $\ell \in [1,a]$, are contained in $\Img(\alpha)$. 
Now suppose $z \in T_2(w_j)$ for some $j > a$; that is, $S_j = \emptyset$. 
Then, the lowest ancestor of $z$ from $\Img(\alpha)$, denoted by $z^F$, is necessarily $z^F = w$. 
By (F-2) of our procedure, we have $\newdepth(w_j) \le 2\mydelta - (\hat{h}_i -\hat{h}_{i-1})$. It then follows that: 
\begin{align*}
g(z^F) - g(z) &= g(w) - g(z) = g(w) - g(w_j) + g(w_j) - g(z) \le g(w) - g(w_j) + \newdepth(w_j)\\
&\le \hat{h}_i - \hat{h}_{i-1} + \newdepth(w_j) \le 2\mydelta. 
\end{align*}
The only case left is that $z \in  edge(w_j, w)$ for some $j > a$. 
Again, its lowest ancestor from $\Img(\alpha)$ is $z^F = w$, and 
$$g(z^F) - g(z) = g(w) - g(z) \le g(w) - g(w_j) \le 2\mydelta .$$
(Again, the last inequality follows from (F-2) of our procedure.) 
Hence property {\sf (P3)} holds for the constructed map $\alpha$. 

What remains is to show that {\sf (P2)} also holds for $\alpha$. In particular, we need to show that for any $u_1, u_2 \in \myF(S)$ such that $\alpha(u_1) \myanceq \alpha(u_2)$, we have that $u_1^{2\mydelta} \myanceq u_2^{2\mydelta}$. 

First, if $\alpha(u_1) = w$ (implying that $u_1 \in S$), then this claim follows from the fact that $(S, w)$ is valid -- Indeed, as $f(\LCA(S)) \le h_i + 2\mydelta = f(u_1) + 2\mydelta$, the node $u_1^{2\mydelta}$ thus is the ancestor of $u^{2\mydelta}$ for any point $u \in \myF(S)$. 
So from now on we assume that $\alpha(u_1) \mydesc w$. 
If $\alpha(u_1) \in T_2(w_j)$ for some $j \in [1,a]$, then by construction of the map $\alpha$, we have $\alpha(u_1) = \alpha_j(u_1)$ and $\alpha(u_2) = \alpha_j(u_2)$ for the  \mypartial{$\mydelta$}{} map $\alpha_j: \myF_1(S_j) \to T_2(w_j)$. Thus $u_1^{2\mydelta} \myanceq u_2^{2\mydelta}$ holds as $\alpha_j$ is \mypartial{$\mydelta$}. 

Now assume otherwise, which means that there exists some $j\in [1,a]$ such that $\alpha(u_1) \in (w_j, w)$ where $(w_j, w)$ denote the interior of the edge connecting $w_j$ and its parent $w$, implying that $\alpha(u_2) \in T_2(w_j) \cup (w_j, w)$. 
Since $f(u_1) = g(\alpha(u_1)) - \mydelta$, there exists some $s \in S_j$ such that $u_1$ is from edge $(s, \hat{s})$ with $\hat{s}$ being the parent of $s$ from $S$. 
%See Figure \ref{fig:Lemma3-1}. 
There are two cases:

(a) Suppose $\alpha(u_2) \in (w_j, w)$. By construction of $\alpha$, there must be $t \in S_j$ and $\hat{t} \in S$ such that $u_2 \in (t, \hat{t})$. 
Since $F(S_j, w_j) = 1$, the pair $(S_j, w_j)$ must be valid, meaning that $f(\LCA(S_j)) \le height(S_j) + 2\mydelta$. As $s, t \in S_j$, it then follows that $s^{2\mydelta}= t^{2\mydelta} \myanceq \LCA(S_j)$. See Figure \ref{fig:Lemma3-1} (b). 
Since in the tree $T_1^f$, $u_1^{2\mydelta} \myanceq s^{2\mydelta}$, $u_2^{2\mydelta} \myanceq t^{2\mydelta}$, and $f(u_1^{2\mydelta}) \ge f(u_2^{2\mydelta})$, it must be that 
$u_1^{2\mydelta}\myanceq u_2^{2\mydelta} \myanceq \LCA(S_j)$. Hence property {\sf (P2)} for holds for $u_1$ and $u_2$. 

(b) The second case is that $\alpha(u_2) \in T_2(w_j)$. In this case, note that $\alpha(u_1) \myanceq \alpha(s) \myanceq \alpha(u_2)$ as $\alpha(s)=w_j$ is the only child node of $\alpha(u_1)$ in $\hatT_2^g$. 
We thus obtain that $s^{2\mydelta} \myanceq u_2^{2\mydelta}$ by applying property {\sf (P2)} w.r.t. the \mypartial{$\mydelta$}{} map $\alpha_j: \myF_1(S_j) \to T_2(w_j)$ to the pair of points $s$ and $u_2$.  
As $u_1^{2\mydelta}\myanceq s^{2\mydelta}$, it then follows that $u_1^{2\mydelta}\myanceq u_2^{2\mydelta}$. 

This finishes the proof that property {\sf (P2)} also holds for the newly constructed map $\alpha: \myF_1(S) \to T_2(w)$. 
Putting everything together, we have that $\alpha$ is an \mypartial{$\mydelta$}{} map. 

\item[$\Leftarrow$: ] Now suppose there is an \mypartial{$\mydelta$}{} map $\alpha: \myF_1(S) \to T_2(w)$ for a valid pair $(S,w)$. We aim to show that $F(S,w) = 1$ in this case. 
As $\alpha$ is monotonically continuous, we know $\alpha(\levelC(S)) \subseteq \levelC(w) = \{w_1, \ldots, w_k\}$. 
For each $i\in [1,k]$, set $S_i = \alpha^{-1}(w_i)$ (we set $S_i = \emptyset$ if $w_i \notin \Img(\alpha)$). 
Obviously, $S_1, \ldots, S_k$ obtained this way form a partition of $\levelC(S)$; that is, $\cup_i S_i = \levelC(S)$ and $S_i\cap S_j = \emptyset$. 
Similar to above, assume w.o.l.g. that $S_1, \ldots, S_a$ are non-empty, and $S_{a+1}, \ldots, S_k$ are empty. 
It is easy to see that the restriction of $\alpha$ to each $\alpha_j: \myF(S_j) \to T_2(w_j)$, for $j\in [1,a]$, gives rise to an \mypartial{$\mydelta$}{} map. 
Furthermore, we claim that each $(S_j, w_j)$, $j\in [1,a]$, is a valid pair. 
In particular, we need to show that $\LCA(S_j) \le h + 2\mydelta$ where $h =  h_{i-1}$ is the height ($f$-value) of nodes in $S_j$ (from \superlevel{} $\slone{i-1}$). 
 This follows from property {\sf (P2)} of map $\alpha$ as for any two $u_1, u_2 \in S_j$, 
$\alpha(u_1) = \alpha(u_2) = w_j$, meaning that $u_1^{2\mydelta} = u_2^{2\mydelta}$. Hence $\LCA(S_j)$ must be at height at most $2\mydelta$ above $f(u_1)  = h_{i-1}$. Thus $(S_j, w_j)$ is a valid pair for any $j \in [1, a]$. 
Since $(S_j, w_j)$ is valid, and it is from level ${i-1}$, it then follows from the induction hypothesis that $F(S_j, w_j) = 1$ for $j\in [1,a]$ as there is a \mypartial{$\mydelta$}{} map $\alpha_j: \myF(S_j) \to T_2(w_j)$. This establishes condition (F-1) in our algorithm \DPalg(). 

Finally, consider any $S_j = \emptyset$ (i.e, $j > a$). 
This means that $w_j \in T_2(w) \setminus \Img(\alpha)$, and thus $T_2(w_j) \subseteq T_2(w) \setminus \Img(\alpha)$. On the other hand, since there is no tree nodes of $\hatT_2^g$ between $\sltwo{i-1}$ and $\sltwo{i}$ (and thus between $w_j$ and its ancestor $w$), the lowest ancestor $w_j^F$ of $w_j$ from $\Img(\alpha)$ must be $w$. This implies that for any $y\in |T_2(w_j)|$, $y^F = w$ as well.  
It then follows that for any $y\in |T_2(w_j)|$, 
$g(w) - g(y) \le 2\mydelta$ by property {\sf (P3)} of $\alpha$. 
We thus have: 
\begin{align*}
g(w_j) - g(y) &= (g(w) - g(y)) - (g(w) - g(w_j)) \le 2\mydelta - (g(w) - g(w_j)) \\
\Rightarrow ~~~~
\newdepth(w_j) &= \max_{y\in |T_2(w_j)|} (g(w_j) - g(y)) \le 2\mydelta - (g(w) - g(w_j)) = 2\mydelta - (\widehat{h}_i - \widehat{h}_{i-1}). 
\end{align*}
This shows that condition (F-2) also holds. 
Hence $F(S, w) = 1$ as $(S,w)$ is valid. 
\end{description}
Lemma \ref{lem:partialgood} then follows from the above two directions. 

\section{Proof of Lemma \ref{lem:sizevalidS}}
\label{appendix:lem:sizevalids}

First, we need the following simple result. 

\begin{claim}\label{claim:epsball}
Given $v \in \uTone$ (not necessarily a tree node of $T_1^f$),  
Set $T_v' :=\{ x \in \uTone \mid x \mydesceq v, \text{and}~f(v) - f(x) \ge 2\delta \}$. 
Then there exists $u \in \uTone$ such that $T_v' \subseteq B_\delta(u, T_1^f)$.
This implies that the sum of degrees of all tree nodes from $T_v'$ is bounded by the $\mydelta$-\degbound{} w.r.t. $T_1^f$ and $T_2^g$. 
\end{claim}
\begin{wrapfigure}{r}{0.4\textwidth}
\vspace*{-0.15in}
\centering 
\includegraphics[height=3.5cm]{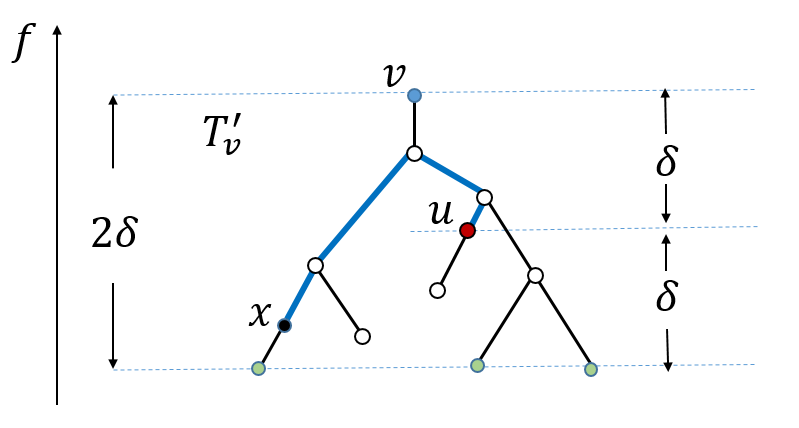}
\vspace*{-0.3in}
\caption{The thickened path is $\pi(u, x)$. \label{fig:deltaball}}
\end{wrapfigure}
%\begin{proof}
\noindent\emph{Proof.~}
See Figure \ref{fig:deltaball} for an illustration: We simply choose $u$ as any descendant $u \mydesceq v$ such that $f(v) - f(u) = \mydelta$. (If there is no point at this height, just take $u$ to be the lowest descendant of $v$.)  
It is easy to see that that for each $x\in |T_v'|$, the path $\pi(u, x)$ is contained inside $T_v'$ and thus all points in this path is within $\mydelta$ height difference from $f(u)$, that is, for any $y\in \pi(u, x)$, $|f(y) - f(u)| \le \delta$. Hence $T_v'  \subseteq B_\mydelta(u; T_1^f)$. 
The bound on the sum of degrees for all tree nodes in $T_v'$ follows from the definition of $\mydelta$-\degbound{}. 
%\end{proof}
\qed{}

\vspace*{0.08in}\noindent Next, to prove Lemma \ref{lem:sizevalidS}, note that points in $S$ are augmented-tree nodes from the augmented tree $\hatT_1^f$ (not necessarily from $T_1^f$), while the \degbound{} parameter $\tau$ is defined w.r.t. the original trees $T_1^f$ and $T_2^g$. 

\begin{wrapfigure}{r}{0.4\textwidth}
\centering 
\includegraphics[height=3cm]{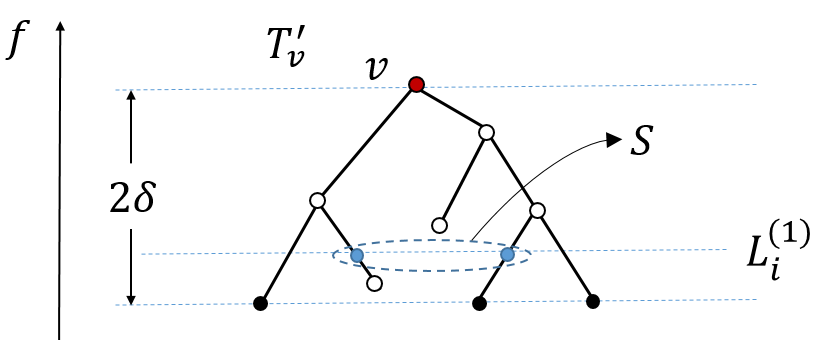}
\caption{Blue dots are $S$ from \superlevel{} $\slone{i}$.\label{fig:deltaS}}
\end{wrapfigure}
Given a valid-pair $(S, w)$, if $|S| = 1$, then the claim holds easily. So we now assume that $|S| > 1$. 
In this case, $(S, w)$ being valid means that $v = \LCA(S)$ is at most $2\mydelta$ height above nodes in $S$. 
Let $T_v' = \{ x \mydesceq v \mid f(v) - f(x) \le 2\mydelta \}$ denote the subset of the subtree $T_1^f(v)$ consisting of all descendant of $v$ in $\uTone$ whose height is within $2\mydelta$ difference to the height $f(v)$ of $v$.
See Figure \ref{fig:deltaS} for an illustration. 
Points in $S$ \emph{may not} be tree nodes from $T'_v$. 
However, as all nodes in $S$ are coming from the same height ($f$-function value), each tree arc in $T'_v$ can give rise to at most one point in $S$. Hence $|S|$ is bounded by the total number of edges in $T'_v$, which in turn is at most the sum of degrees of all tree nodes in $T'_v$. It then follows from Claim \ref{claim:epsball} that $|S| \le \myk$ as claimed.  

Next, we now bound $|\levelC(S)|$, the number of child-nodes of $S$. 
Indeed, first, suppose a point $s\in S$ is an augmented tree node of $\hatT_1^f$ but not a tree node of $T_1^f$. Then $s$ can give rise to only one child-node in $\levelC(S)$, and we can charge this node to the tree arc $s$ lies in. 
Otherwise, suppose $s$ is also a tree node in $T_1^f$. Then the number of its child-nodes is already counted when we compute the sum of degrees of all tree nodes in $T'_v$. 
Putting these two together, we have that $|\levelC(S)|$ is also bounded from above by the sum of degrees of all tree nodes within $T'_v$, which is further bounded by $\myk_\mydelta(T_1^f, T_2^g) =\myk$ by Claim \ref{claim:epsball}. 
This finishes the proof for Lemma \ref{lem:sizevalidS}. 
%\qed{}
%\end{proof}

\section{Proof of Lemma \ref{lem:coarsevalidbound}}
\label{appendix:lem:coarsevalidbound}

In what follows we will separate valid pairs to two classes: (i) a \emph{singleton-pair} $(S, w)$ is a valid pair with $|S|=1$, or (ii) a \emph{non-singleton-pair} is a valid pair $(S, w)$ with $|S| > 1$. 

First, consider singleton-pairs, which have the form $(s, w)$ 
with $s \in V(\hatT_1^f)$ and $w\in V(\hatT_2^g)$. 
The number of augmented tree nodes in each augmented tree $\hatT_1^f$ or $\hatT_2^g$ is $O(n^2)$. 
%, and the number of augmented tree nodes within any \superlevel{} is at most $O(n)$. 
Hence there are $O(n^2)$ choices of $w$. For each $w \in \sltwo{i}$, it can be paired with $O(n)$ potential augmented-tree nodes from the \superlevel{} $\slone{i}$ of $\hatT_1^f$. 
Therefore the total number of singleton-pairs is bounded by $O(n^3)$. 

Next we bound the number of non-singleton-pairs $(S, w)$, with $|S| > 1$, that Algorithm \DPalg() may inspect. 
Given such a set $S \subset V(\hatT_1^f)$ from the \superlevel{} $\slone{i}$, its common ancestor $v = LCA(S)$ has to be a \emph{tree node} in $V(T_1^f)$ whose height ($f$-function value) is at most $2\delta$ above points in $S$; that is, $f(v) \le h_i + 2\mydelta$ where recall that $h_i$ is the height ($f$-value) of \superlevel{} $\slone{i}$. 
As $v \in V(T_1^f)$, there are $|V(T_1^f)| \le n$ choices for $v$. 
We now count how many possible sets of $S$ a fixed choice $v\in V(T_1^f)$ can produce. 

To this end, set $T_v' = \{x \mydesceq v \mid f(v) - f(x) \le 2\mydelta\}$ as in Claim \ref{claim:epsball}, by which we know that the sum of degrees of all nodes within $T'_v$ is $\myk$. 
Hence the total number tree edges (from $T_1^f$) contained in $T_v'$ is at most $\myk$. 
On the other hand, there can be $O(n)$ number of \superlevel{}s intersecting $T_v'$. For each such \superlevel{}, the number of augmented tree nodes contained in $|T_v'|$ is bounded by the number of tree edges of $T_1^f$ in $T_v'$ and thus by $\myk$. 
It then follows that for each \superlevel{} intersecting $T_v'$, the number of potential subset $S$'s it can produce is at most $2^\myk$. 
Hence all \superlevel{}s from $T_v'$ can produce at most $O(n 2^\myk)$ number of potential sets of $S$. 
Overall, considering all $O(n)$ choices of $v$'s, there can be at most $O(n^2 2^\myk)$ potential $S$'s that the algorithm will never need to inspect.  

For each potential $S$, say from $\slone{i}$, there are $n$ choices for $w$ (as it must be an augmented-tree node from the \superlevel{} $\sltwo{i}$ of $\hatT_2^g$). 
Putting everything together, we have that there are at most $O(n^3 2^\myk)$ number of valid-pairs $(S, w)$ with $|S| > 1$, and they can also be enumerated within the same amount of time.

\section{A Faster FPT Algorithm for Deciding ``Is $d_I(T_1^f, T_2^g) \le \mydelta$''}
\label{appendix:subsec:fasteralg}

In what follows, we first bound the number of \sensiblepair{}s in Lemma \ref{lem:sensiblepairs}. We then show that Algorithm \DPalg() can be modified to consider ony \sensiblepair{}s, and achieves the claimed time complexity. 

\paragraph{\Edgelistpair{}s.}
Consider any pair $(S, w)$ with $S \subseteq \slone{i}$ and $w \in \sltwo{i}$. Suppose $S = \{s_1, \ldots, s_\ell\}$; each $s_j$ is an augmented-tree node from some tree arc, say $e_j$ of $T_1^f$. We call $A = \{e_1, \ldots, e_\ell\} \subseteq E(T_1^f)$ the \emph{edge-list supporting $S$}. Let $\alpha \in E(T_2^g)$ be the tree edge in $T_2^g$ such that $w \in \alpha$. 
We say that the \emph{\edgelistpair{}} ($A, \alpha$) \emph{supports} $(S,w)$. 
Two different pairs $(S,w)$ and $(S', w')$ could be supported by the same \edgelistpair{}. However, we claim that each \edgelistpair{} can support at most $4$ sensible-sets. 
Recall that $E(T)$ stands for the edge set of a tree $T$. 
Given an arbitrary tree edge $e = (u, u') \in E(T_1^f)$, we refer to the endpoint, say $u$, with smaller $f$-value as the \emph{lower-endpoint of $e$}, while the other one with higher $f$-value as the \emph{upper-endpoint of $e$}.
Similarly define the lower/upper-endpoints for edges in $T_2^g$. 

\begin{figure}[htbp]
\centering 
\includegraphics[height=4cm]{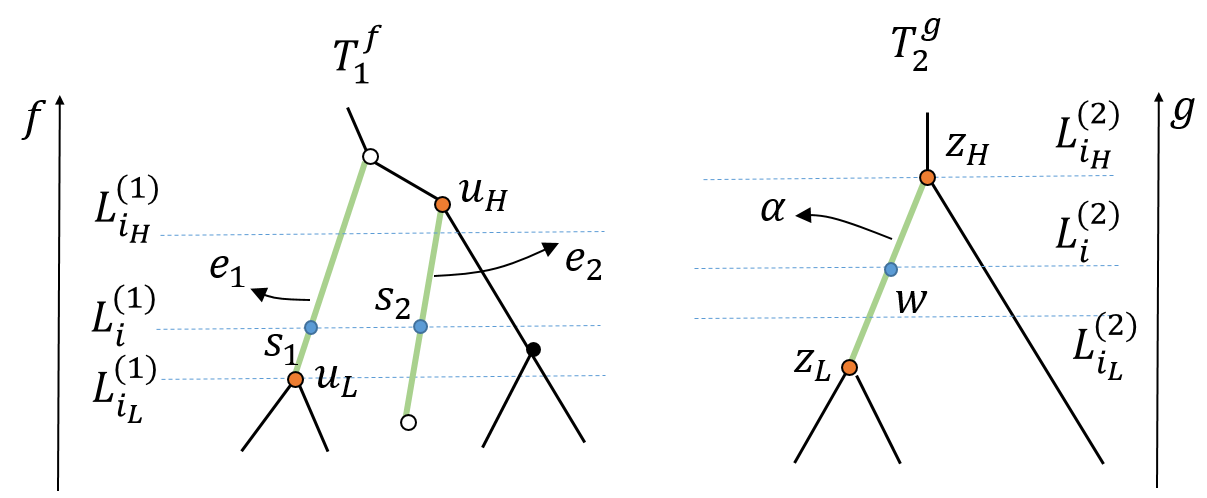}
\vspace*{-0.15in}
\caption{Green edges are $A = \{e_1, e_2\}$ and $\alpha$. In this example, $i_L$ is the index of the \superlevel{} passing through $u_L$, while $i_H$ is the index of the \superlevel{} passing through $z_H$. \label{fig:edgelist}}
\end{figure}
\begin{lemma}\label{lem:edgelist}
Each edge-list $(A, \alpha)$, with $A \subseteq E(T_1^f)$ and $\alpha \in E(T_2^g)$, can support at most $4$ \sensiblepair{}s.  
\end{lemma}
\begin{proof}
Let $u_H$ be the lowest upper-end point of all edges in the edge-list $A = \{ e_1, \ldots, e_\ell\}$ and $u_L$ be the highest lower-endpoint of all edges in it. 
Suppose $\alpha = (z_L, z_H)$ with $g(z_L) < g(z_H)$. 
See Figure \ref{fig:edgelist} for an illustration. 
Let $i_H$ be the \emph{smaller} index of the \superlevel{}s supporting $u_H$ from $\setSL_1$ and supporting $z_H$ from $\setSL_2$, respectively. 
Similarly, let $i_L$ be the \emph{larger} index of the \superlevel{}s supporting $u_L$ from $\setSL_1$ and supporting $z_L$ from $\setSL_2$. 
Then any valid pair ($S,w$) supported by $(A,w)$ can only come from the $i$-th \superlevel{} with $i \in [i_L, i_H]$. 
Furthermore, any {\bf \sensiblepair{}} $(S,w)$ supported by $(A,w)$ must be from the $i$-th \superlevel{}s with $i \in \{ i_L, i_L+1, i_H-1, i_H\}$, as for any other choice of $i \in [i_L, i_H]$, no point from $S$, $w$, their children or parents in the augmented trees $\hatT_1^f$ and $\hatT_2^g$, can be a tree node. 
This proves the lemma. 
\end{proof}

\begin{lemma}\label{lem:sensiblepairs}
There are $O(n^2 2^\myk)$ distinct \edgelistpair{}s supporting \sensiblepair{}s. 
Hence there are $O(n^2 2^\myk)$ \sensiblepair{}s, and they can be computed in $O(n^2 2^\myk)$ time. 
\end{lemma}
\begin{proof}
First, it is easy to see that there are $O(n^2)$ singleton \sensiblepair{}s. 
Indeed, each singleton pair $(\{s\},w)$ is necessarily supported by a singleton \edgelistpair{} of the form $(e, \alpha)$ with $s\in e$ and $w\in \alpha$. 
By Lemma \ref{lem:edgelist}, each \edgelistpair{} can support at most $4$ \sensiblepair{}s (thus at most $4$ singleton \sensiblepair{}s). Since there are $O(n^2)$ number of singleton \edgelistpair{}s, it then follows that there can be $O(n^2)$ singleton \sensiblepair{}s. 

We now focus on non-singleton \sensiblepair{}s. 
To this end, we will distinguish two types of \sensiblepair{}s and bound them separately. 
A \sensiblepair{} $(S,w)$ is \emph{type-1} if condition (C-1) in Definition \ref{def:sensiblepair} holds; and \emph{type-2} if condition (C-2) holds. 

First, we bound type-1 \sensiblepair{}s. We say that a set $S$ from some \superlevel{} $\slone{i}$ is a \emph{sensible-set} if $S$ satisfies condition (C-1). (In other words, if $(S,w)$ is type-1 \sensiblepair{}, then $S$ must be a sensible-set.) 
We will now bound the number of sensible-sets. 
Given a non-singleton type-1 \sensiblepair{} $(S, w)$ (thus $|S| > 1$),  similar to the argument in the proof of Lemma \ref{lem:coarsevalidbound}, the lowest common ancestor $v = LCA(S)$ of augmented-tree nodes in $S$ must be a \emph{tree node} from $V(T_1^f)$ less than $2\mydelta$ heighb above $S$. 
Set $T_v' = \{x \mydesceq v \mid f(v) - f(x) \le 2\mydelta\}$: we know that $S$ is contained within $T_v'$. 
By Claim \ref{claim:epsball}, both the number of tree nodes and the number of tree arcs within $T_v'$ are bounded from above by $\myk$. 
Since there are only $O(\myk)$ number of tree-arcs of $T_1^f$ contained in $T_v'$, $T_v'$ can produce at most $2^\myk$ distinct edge-lists. 
On the other hand, a similar argument as the one used for Lemma \ref{lem:edgelist} can also show that each edge-list $A \subset E(T_1^f)$ can only support at most $4$ sensible-sets. 
Hence there are at most $O(2^\myk)$ sensible-sets possible from $T_v'$. 
Ranging over all $O(n)$ choices of $v$'s, there can be $O(n 2^\myk)$ possible sensible-sets. 

For each sensible-set $S$, there can be at most $n$ choices of $w$ forming a potential \sensiblepair{} $(S, w)$ with it (where $w$ has to be an augmented tree-nodes from a specific \superlevel{} in $\hatT_2^g$). 
This leads to $O(n^2 2^\myk)$ type-1 \sensiblepair{}s in total, and they can be computed within the same time complexity.

Next, we bound the type-2 \sensiblepair{}s ($S,w$). 
Note that there are only $O(n)$ choices of $w$ participating type-2 \sensiblepair{}s (as each tree edge in $T_2^g$ can give rise to at most $4$ choices of $w$'s satisfying condition (C-2) in Definition \ref{def:sensiblepair}). 
For each fixed choice $w \in \alpha_w$, by an argument similar to the one used for the type-1 case, we can argue that there are only $O(n 2^\myk)$ number of edge-lists from $E(T_1^f)$ supporting some \sensiblepair{} of the form $(S, w)$. 
Ranging over $O(n)$ choices for $w$, this gives rise to $O(n^2 2^\myk)$ total \edgelistpair{}s supporting at most $O(n^2 2^\myk)$ type-2 \sensiblepair{}s. 
This finishes the proof of the claim. 
\end{proof}

We now modify Algorithm \DPalg($T_1^f, T_2^g, \mydelta$) so that it will only compute feasibility $F(S, w)$ for \sensiblepair{} $(S,w)$'s. 
\begin{description}\denselist
\item[Description of \modifyDPalg($T_1^f, T_2^g, \mydelta$)] 
\item[] (1) Compute all \edgelistpair{}s $\Xi$ that support sensible-sets. For each edge-list $(A, \alpha) \in \Xi$, associate to it all (at most 4) sensible-sets that $(A, \alpha)$ supports. 
 From $\Xi$, compute all \sensiblepair{} and associate, to each \superlevel{} $\slone{i}$, the collection of \sensiblepair{}s $(S,w)$ such that $S \subseteq \slone{i}$. 
\item[] (2) Compute the feasibility $F(S, w)$ in a bottom-up manner (i.e, in increasing order of their \superlevel{} indices) only for \sensiblepair{}s $(S,w)$. Consider a \sensiblepair{} ($S,w)$ from the $i$-th \superlevel{}s, that is, $S \subseteq \slone{i}$ and $w \in \sltwo{i}$. 
\begin{description}\denselist
\item[Base case $i=1$:]
All valid-pairs from the first \superlevel{} are \sensiblepair{}s. The computation of $F(S,w)$ is the same as in the Base-case of Algorithm \DPalg().
\item[When $i > 1$:] We have already computed the feasibility for all \sensiblepair{}s from level ($i-1$) or lower. 
Now given a valid pair $(S,w)$ from level-$i$, we will use a similar strategy as in Algorithm \DPalg(). 
Specifically, let $\levelC(S) \subseteq \slone{i-1}$ denote the set of children of $S$ from \superlevel{} $\slone{i-1}$, and $\levelC(w) = \{w_1, \ldots, w_k\} \subseteq \sltwo{i-1}$ the set of children of augmented-tree node $w$. 

If $\levelC(w)$ is empty, the $F(S,w)=1$ if and only if $\levelC(S) = \emptyset$. 

If $\levelC(w)$ is not empty, then we check whether there exists a partition of $\levelC(S) = S_1\cup \cdots \cup S_k$ such that conditions (F-1) and (F-2) in Algorithm \DPalg() hold. 

However, we need to modify condition (F-1), as it is possible that $(S_j, w_j)$ is valid but not a \sensiblepair{} and thus $F(S_j, w_j)$ has not yet being computed.

More precisely, we do the following: if $S_j$ is not valid, then obviously $F(S_j, w_j) = 0$. 
Otherwise, $(S_j, w_j)$ is a valid-pair. If it is also a \sensiblepair{}, then $F(S_j, w_j)$ is already computed. 

The remaining case is that $(S_j, w_j)$ is valid but not a \sensiblepair{}. 
Let $A$ be the edge-list supporting $S_j$, and $\alpha\in E(T_2^g)$ the edge containing $w_j$. 
Let $(S', w')$ be the highest \sensiblepair{} supported by $(A, \alpha)$ but from a \superlevel{} below that of $(S_j, w_j)$. 
If such a ($S', w'$) does not exist, then we set $F(S_j, w_j) = 0$. 
Otherwise, set $F(S_j, w_j) = F(S', w')$. 

The {\sf (modified F-1)} is: If $S_j \neq 0$, the $F(S_j, w_j)$ as setup above should equal to $1$. 
\end{description}
\item[Output:] The algorithm returns ``yes" if and only if $F(root(\hatT_1^f), root(\hatT_2^g)) = 1$. 
\end{description}

%\begin{theorem}\label{thm:modifiedcorrect}
%Algorithm \modifyDPalg($T_1^f, T_2^g, \mydelta$) returns ``yes'' if and only if $d_I(T_1^f, T_2^g) \le \mydelta$. 
%\end{theorem}
%\begin{proof}
\paragraph{Proof of Part (i) of Theorem \ref{thm:fasteralg}.} 
We now show that Algorithm \modifyDPalg($T_1^f, T_2^g, \mydelta$) returns ``yes'' if and only if $d_I(T_1^f, T_2^g) \le \mydelta$. 

First, note that $(root(\hatT_1^f), root(\hatT_2^g))$ must be a \sensiblepair{}, and thus will be computed by our algorithm. 
We now show that for any \sensiblepair{} ($S,w$), the feasibility $F(S,w)$ computed by \modifyDPalg() is the same as that computed by \DPalg(). 
To differentiate the two feasibility values, we use $\Fnew(S,w)$ and $\Fold(S, w)$ to denote the feasibility computed by \modifyDPalg() and by \DPalg(), respectively. 

We prove this by induction w.r.t. \sensiblepair{}s from \superlevel{}s of increasing indices. 
At the base case when the index of \superlevel{} $i=1$, it is easy to check that a valid pair has to be sensible -- in fact, $S$ and $w$ are either empty or contains only tree-nodes from $T_1^f$ and $T_2^g$. 

Now consider a \sensiblepair{} $(S, w)$ from the $i$-th \superlevel{}. 
Let $\levelC(S)$ and $\levelC(w)$ be as defined in \modifyDPalg() (which is the same as in \DPalg()). 
For any partition $\levelC(S) = S_1 \cup \cdots S_k$ of $S$, we only need to show that condition (F-1) holds if and only if condition (modified F-1) holds. Furthermore, the only case we need to consider is when $S_j$ is valid but the pair $(S_j, w_j)$ is not sensible. 
Note that in this case, $\Fnew(S_j, w_j)$ could be set to $\Fnew(S', w')$ as described in \modifyDPalg(). 
As in the algorithm, let $A$ denote the edge-list supporting $S_j$, and $\alpha = (z_L, z_H)$, with $g(z_L) \le g(z_H)$, is the edge from $T_2^g$ containing $w_j$. 
We now prove the following two claims:  

\vspace*{0.15in}\emph{(Claim-1) First we show that if $\Fnew(S_j, w_j) = 1$; then it must be that $\Fold(S_j, w_j) = 1$.}
In this case, $(S', w')$ must exist when running Algorithm \modifyDPalg(), and it is the highest \sensiblepair{} supported by \edgelistpair{} $(A, \alpha)$ but below $S_j$; with $\Fnew(S', w') = 1$. 
We claim that all augmented-tree nodes in $S'$ and augmented-tree node $w'$ must all be of degree-2 (i.e, they are all in the interior of some tree arcs of the original trees $T_1^f$ and $T_2^g$, and none of them is a tree node for $T_1^f$ or $T_2$). 

Indeed, suppose this is not the case and some augmented-tree nodes in $S'$ or $w'$ is in fact a tree node. 
Then, as $(S', w')$ is supported by $(A, \alpha)$, by the proof of Lemma \ref{lem:edgelist}, the only possibility is that $w' = z_L$ or $S'$ contains $u_L$, where $u_L$ is the {\bf highest} lower-endpoint of all edges in $A$.
Suppose $S'$ is from \superlevel{} $\slone{c}$. 
Then consider its parents $S''$ from \superlevel{} $\slone{c+1}$ as well as the parent $w''$ of $w'$ from $\sltwo{c+1}$. 
Since $(S',w')$ (coming from $c$-th \superlevel) is below $(S_j, w_j)$ which comes from the ($i-1$)-th \superlevel{}, we have that $c+1 \le i-1$. 
%As $c+1 \le i-1$ (recall that $S_j$ is from \superlevel{} $\slone{i-1}$), 
It then follows that $S''$ must be supported by the same edge-list $A$, and similarly, $w' \prec w'' \preceq w_j$ along edge $\alpha = (z_L, z_H)$. 
As $S'$ is valid, $S''$ must be valid. As $S'$ and $w'$ contain a tree node of $T_1^f$ or $T_2^g$, $(S'', w'')$ must be sensible (satisfying either condition (C-1) or (C-2) in Definition \ref{def:sensiblepair}). This however contradicts that $(S', w')$ is the highest \sensiblepair{} below $(S_j, w_j)$ supported by \edgelistpair{} $(A, \alpha)$. 
Hence the assumption is wrong and all points in $S'$ and $w'$ must be in the interior of some tree arcs. 

Since $(S', w')$ is a \sensiblepair{} and is from a lower \superlevel{} than $(S, w)$, by the induction hypothesis, we already have that $\Fold(S', w') = \Fnew(S', w')$, meaning that $\Fold(S', w') = 1$. 
Now during the execution of Algorithm \DPalg(), it will inspect a sequence of {\bf valid pairs} 
$$(S', w') = (S_j^{(0)}, w_j^{(0)}), (S_j^{(1)}, w_j^{(1)}), \ldots, (S_j^{(t)}, w_j^{(t)}) = (S_j, w_j), $$ 
listed in increasing heights, supported by \edgelistpair{} $(A, \alpha)$. 
As $S'$ is valid, and all $S_j^{(b)}$'s, for $b\in [1, t]$, are supported on the same edge-list, each $(S_j^{(b)}, w_j^{(b)})$ must be valid as well. Furthermore, as $w_j^{(b)}$ is the only child of $w_j^{(b+1)}$ for any $b\in [0, t-1]$ (they are both contained in the interior of edge $\alpha$), 
it then follows that during the execution of Algorithm \DPalg(), $\Fold(S_j^{(b+1)}, w_j^{(b+1)})$ is set to be $\Fold(S_j^{(b)}, w_j^{(b)})$ for each $b\in [0, t-1]$. 
Hence 
$$(\Fold(S_j, w_j) =)~ \Fold(S_j^{(t)}, w_j^{(t)}) = \cdots = \Fold(S_j^{(1)}, w_j^{(1)}) = \Fold(S_j^{(0)}, w_j^{(0)}) ~(= \Fold(S', w')).$$
Since $\Fold(S', w') = 1$, it thus follows that $\Fold(S_j, w_j) = 1$, establising (Claim-1).

\vspace*{0.15in}\emph{(Claim-2) Next we show that if $\Fold(S_j, w_j)=1$, then $\Fnew(S_j,w_j) = 1$.} 
In this case, let $(S_j^{(0)}, w_j^{(0)}), (S_j^{(1)}, w_j^{(1)}), \ldots, (S_j^{(t)}, w_j^{(t)}) = (S_j, w_j)$ denote the sequence of all pairs (not necessarily valid) which are: (i) supported by the \edgelistpair{} $(A,\alpha)$, (ii) at or below $(S_j, w_j)$; and (iii) strictly higher than the \superlevel{}s containing $u_L$ and $z_L$. We also assume that this sequence is listed in increasing heights. 

Assume that the higher \superlevel{} containing either $u_L$ or $z_L$ is the $c$-th \superlevel{}. 
Then (iii) above implies that (1) all (augmented-tree) nodes in $(S_j^{(b)}, w_j^{(b)})$, for all $b\in [0,t]$, are degree-2 nodes (i.e, not tree nodes); (2) $(S_j^{(0)}, w_j^{(0)})$ must come from the ($c+1$)-th \superlevel{}s $\slone{c+1}$ and $\sltwo{c+1}$; and (3) no $(S_j^{(b)}, w_j^{(b)})$, for $b\in [1, t-1]$, can be sensible. 
Note that (2) implies that $(S_j^{(0)}, w_j^{(0)})$ satisfies either condition (C-1) or (C-2) of Definition \ref{def:sensiblepair} (although we have not yet shown it is valid, and thus we do not know whether it is sensible or not yet). 

We argue that if $\Fold(S_j,w_j) = 1$, then $\Fold(S_j^{(b)}, w_j^{(b)}) = 1$ for every $b \in [0, t)$. This is because in Algorithm \DPalg(), $\Fold(S_j^{(b)}, w_j^{(b)})$ is set to be $\Fold(S_j^{(b-1)}, w_j^{(b-1)})$ for each $b\in [1,t]$, as $w_j^{(b-1)}$ is the only child of $w_j^{(b)}$ for $b\in [1,t]$. 
It then follows that $\Fold(S_j^{(0)}, w_j^{(0)}) = 1$, meaning $(S_j^{(0)}, w_j^{(0)})$ must be a valid pair. 
Combined with the fact that $(S_j^{(0)}, w_j^{(0)})$ satisfies either condition (C-1) or (C-2) of Definition \ref{def:sensiblepair} as shown in the previous paragraph, we have that $(S_j^{(0)}, w_j^{(0)})$ is a \sensiblepair{}. Furthermore, by (3) from the previous paragraph, $(S_j^{(0)}, w_j^{(0)})$ has to be the highest \sensiblepair{} supported by \edgelistpair{} $(A, \alpha)$. 
Hence Alorithm \modifyDPalg() will set $(S',w')=(S_j^{(0)}, w_j^{(0)})$. By our induction hypothesis, we have $\Fnew(S',w') = \Fold(S', w')$, implying that the modified algorithm will set $\Fnew(S_j,w_j) = \Fnew(S',w') = \Fold(S',w') = 1$. This proves (Claim-2). 

Combining (Claim-1) and (Claim-2) above, we have that for any \sensiblepair{} $(S,w)$, $\Fnew(S,w) = \Fold(S,w)$. It then follows that Algorithm \modifyDPalg() reutnrs the same answer as Algorithm \DPalg(), which, combined with part (i) of Theorem \ref{thm:DPalg}, establishes the correctness of Algorithm \modifyDPalg(). 

\paragraph{Proof of Part (ii) of Theorem \ref{thm:fasteralg}}
%\begin{theorem}\label{thm:fastertime}
%Algorithm \modifyDPalg($T_1^f, T_2^g, \mydelta$) can be implemented to run in $O(n^2 2^\myk \myk^{\myk+2}\log n)$ time, where $n$ is the total complexity of input trees $T_1^f$ and $T_2^g$, and $\myk = \myk_\mydelta(T_1^f, T_2^g)$ is the $\mydelta$-\degbound{} w.r.t. $T_1^f$ and $T_2^g$. 
%\end{theorem}
%\begin{proof}
We now show that Algorithm \modifyDPalg() can be implemented to run in the claimed time complexity. 

First, Lemma \ref{lem:sensiblepairs} shows that we can compute all \sensiblepair{}s, as well as the set of \edgelistpair{}s $\Xi$ supporting them, in time $O(n^2 2^\myk)$ time. 
We store all \sensiblepair{}s in increasing order of their \superlevel{}s, and process it one by one. 
For each edge-list pair $(A, \alpha)$ with $A \subseteq E(T_1^f)$ and $\alpha \in E(T_2^g)$, we link to it the four \sensiblepair{}s it supports. 
For each \sensiblepair{}, it also stores a pointer, linking to the edge-list pair supporting it. 

In step (2) when computing $F(S,w)$, we need to be able to search whether a pair $(S_j, w_j)$, with $S_j \subseteq \levelC(S)$ and $w_j \in \levelC(w)$, is sensible or not, and identify $(S',w')$ through the \edgelistpair{} supporting ($S_j, w_j$) when  necessary. 
This can be done by storing all edge-list pairs in $\Xi$ in some data structure that supports search efficiently. 
To this end, we view each edge-list pair $(A,\alpha)$ as a set of ordered edge indices [$id_1, \ldots, id_t$], where $id_1, \ldots, id_{t-1}$ are indices for edges in  $A$ from $T_1^f$ in increasing order, while $id_t$ is index for edge $\alpha$ in $T_2^g$. 
Given the collection of all edge-list pairs $\Xi$ that support some \sensiblepair{}s, we first use a standard balanced binary search tree $\Pi$ to store all indices occured in the first position of the ordered index-sets representation of \edgelistpair{} in $\Xi$.  
Next, for each index, say $id_1$ stored in this tree $\Pi$, we then associate to it a secondary data structure (another balanced binary search tree $\Pi_{id_1}$) to store the set of second indices from \edgelistpair{}s of the form $[id_1, \ldots]$. We then repeat, and build a third-level tree for each index in the secondary tree $\Pi_{id_1}$, say $id_2$, to store all indices in the 3rd position from \edgelistpair{}s of the form $[id_1, id_2, \ldots]$.  
Let $\ell$ be the largest cardinality of any \edgelistpair{}. Then this process is repeated at most $\ell$ times to build auxiliary trees of at most $\ell$ levels. Each auxiliary tree has size at most $n$ as the choice of each position is bounded by the total number of edges in each input tree. 

By Lemma \ref{lem:sizevalidS}, $|S| \le \myk$ for any valid pair ($S,w$). Hence $\ell \le \myk+1$. 
This data structure takes $O(\ell \cdot |\Xi|) = O(n^2 \myk 2^\myk)$ space, since each index of every \edgelistpair{} will only be stored at most once. 
The data structure can be built in $O(n^2 \myk 2^\myk \log n)$ time. 
Finally, to search for a specific \edgelistpair{} (represented by the ordered-index set) $[id_1, \ldots, id_t]$, we need to perform $t$ searches in $t$ balanced binary search trees each of which is of size at most $n$. 
Hence a search for an \edgelistpair{} takes $O(t \log n) = O(\myk \log n)$ time. 

Finally, we analyze the time complexity to compute $F(S,w)$ for a fixed \sensiblepair{} $(S,w)$ in step (2) of Algorithm \modifyDPalg().
First, observe that $k=|\levelC(w)| = 
degree(w) \le \myk$, and $|\levelC(S)| \le \myk$ by Lemma \ref{lem:sizevalidS}. Hence  the number of partitioning of $\levelC(S)$ our 
algorithm will check is bounded by 
$O(|\levelC(S)|^{k}) = O(\myk^\myk)$. 
For each partition $S=S_1 \cup \cdots S_k$, and for each $j\in [1, k]$, 
it takes $O(\myk \log n)$ time to search for $(S', w')$ as required in (modified F-1). Hence overall, it takes $O(k \myk \log n) = O(\myk^2 \log n)$ to check conditions (modified F-1) or (F-2) for all $j\in [1,k]$. 
Putting everything 
together, we have that the time complexity of 
our algorithm is bounded from above by 
$O({n^2}{2^{\myk}}\myk^{\myk+2} \log n)$, as claimed. 
This finishes the proof of Part (ii) of Theorem \ref{thm:fasteralg}. 

\section{Proof of Lemma \ref{lem:criSet}}
\label{appendix:lem:criSet}

Imagine we run the dynamic programming algorithm \DPalg($T_1^f, T_2^g, \Intdopt$) as introduced in the previous section. It should return `yes'. 
In particular, consider the augmented trees $\hatT_1^f$ and $\hatT_2^g$ we construct w.r.t parameter $\Intdopt$ as in Section \ref{subsec:FPT}, where $\hatT_1^f$ (resp. $\hatT_2^g$) is augmented from $T_1^f$ (resp. $T_2^g$) by including all points in the \superlevel{}s also as tree nodes. 
Algorithm \DPalg($T_1^f, T_2^g, \Intdopt$) will compute the feasibility $F(S, w)$ for all valid pairs $(S,w)$ in a bottom-up manner, where $S$ and $w$ are from those \superlevel{}s $\slone{i}$ and $\sltwo{i}$ respectively. In the end, it will return $F(root(\hatT_1^f), root(\hatT_2^g))$, which should be equal to `1' in this case. 

Now imagine that we will decrease $\Intdopt$ by an infinitesimally small quantity $\nu>0$ to $
\delta' = \Intdopt - \nu$. 
Algorithm \DPalg($T_1^f, T_2^g, \delta'$) should return `no' since $\Intdopt$ is the smallest possible value on which the dynamic programming algorithm returns `yes'. 
Recall that the \superlevel{}s w.r.t. a generic $\mydelta$ value are defined as follows: 
\begin{align*}
\setSL_{1} &:= \{ \myL(c) \mid c \in \myC_1 \} ~\cup ~\{ \myL(c-\mydelta) \mid c \in \myC_2 \} ~~~\text{and} \\
\setSL_{2} &:= \{ \myL(c+\mydelta) \mid c \in \myC_1 \} ~\cup~ \{ \myL(c) \mid c \in \myC_2 \},  
\end{align*}
where $\myC_1$ (resp. $\myC_2$) contains the $f$-function values of tree nodes in $T_1^f$ (resp. $g$-values of tree nodes in $T_2^g$). 
Half of these \superlevel{}s are independent of the value of $\mydelta$; these are levels in $\{\myL(c) \mid c\in \myC_1\} \subset \setSL_{1}$ for $T_1^f$ and $\{ \myL(c) \mid c \in \myC_2 \} \subset \setSL_{2}$ for $T_2^t$, which pass through tree nodes in $T_1^f$ and $T_2^g$, respectively. We call these \superlevel{}s as \emph{fixed \superlevel{}s}. 
The other half are \emph{induced \superlevel{}s} and they are at height either $\mydelta$ above or below some fixed \superlevel{}s. 

Now consider the \superlevel{}s w.r.t. $\Intdopt$. 
If any induced \superlevel{}, say $\slone{i} \in \setSL_1$, also passes through a tree node, say $v$, in $T_1^f$. 
This means that $\Intdopt \in \criSet_1$: Indeed, since $\slone{i}$ is induced, it is of the form 
$\slone{i} = L(c-\Intdopt)$ with $c = g(w) \in \myC_2$ for some $w \in V(T_2^g)$. Then $\Intdopt = g(w) - f(v)$ with $v\in V(T_1^f)$ and $w\in V(T_2^g)$. Thus $\Intdopt \in \criSet_1$ and the lemma holds. 

Hence from now on we assume that no induced \superlevel{} w.r.t $\Intdopt$ passes through any tree node in $T_1^f$ nor in $T_2^g$. 
As we decrease the value of $\Intdopt$ slightly to $\delta'$, the induced \superlevel{}s move (up or down accordingly) by the same amount $\nu = \Intdopt - \delta'$. 
Since no induced \superlevel{} contains any tree node in $T_1^f$ or $T_2^g$, for any induced \superlevel{} $L$, each point in it only moves up or down in the interior of some tree edge of $T_1^f$ or $T_2^g$, and no new point can appear in a \superlevel. Hence there is a canonical bijection between the tree nodes $V(\hatT_1^f(\Intdopt))$ (resp. $V(\hatT_2^g(\Intdopt))$ ) in the augmented tree $\hatT_1^f(\Intdopt)$ w.r.t. $\Intdopt$, 
and $V(\hatT_1^f(\delta'))$ (resp. $V(\hatT_2^g(\delta'))$) w.r.t. $\delta'$. 
(This bijection also induces a bijection of points (augmented tree nodes) within the $i$-th \superlevel{} w.r.t $\Intdopt$ and points in corresponding $i$-th \superlevel{} w.r.t $\delta'$.) 
Under this bijection, we use $S(\delta')$ to denote the set of tree nodes in $\hatT_1^f(\delta')$ (resp. in $\hatT_2^g(\delta')$) corresponding to $S \subset V(\hatT_1^f(\Intdopt))$ (resp. $S\subset V(\hatT_2^g(\Intdopt))$). 

Since $F(root(\hatT_1^f(\delta')), root(\hatT_2^g(\delta'))) = `no'$, it means that during the bottom-up process of algorithm \DPalg($T_1^f,T_2^g, \delta'$), 
%to compute $F_{\delta'}(S(\delta'),w(\delta'))$s, 
at some point  $F_{\Intdopt}(S,w)=1$ (w.r.t. parameter $\Intdopt$), yet $F_{\delta'}(S(\delta'),w(\delta')) = 0$ (w.r.t. parameter $\delta'$). 
Let $(S, w)$ be such a pair from the lowest \superlevel{} that this happens.  

This could be caused by the following reasons: 
\begin{itemize}\denselist
\item[(1)] $(S(\delta'), w(\delta'))$ is no longer a valid pair. This means that $u = LCA(S)$ (which is necessarily a tree node in $V(T_1^f)$) must be at height $h + 2\Intdopt$, where $h = f(S)$. 
If $S$ is from a fixed \superlevel, then there exists a tree node $v\in V(T_1^f)$ such that $f(v) = h$. It then follows that $2\Intdopt = f(u) - f(v)$, implying that $\Intdopt \in \criSet_2$. 
Otherwise, $S$ is from an induced \superlevel{} $\slone{i}$. This means that there exists a tree node $w'\in V(T_2^g)$ such that $h = g(w') - \Intdopt$. It then follows that $f(u) - g(w') = h+2\Intdopt - (h+\Intdopt) = \Intdopt$, implying that $\Intdopt \in \criSet_1$. 
\item[(2)]  $(S(\delta'), w(\delta'))$ is still valid, $F_\Intdopt(S, w)  = 1$ but $F_{\delta'}(S(\delta'),w(\delta')) = 0$. 
There are two cases how $F_{\delta'}(S(\delta'),w(\delta')) = 0$ in algorithm \DPalg($T_1^f,T_2^g, \delta'$).  
\begin{itemize}\denselist
\item[(2.a)]
$\levelC(S(\delta')) \neq \emptyset$ but $\levelC(w(\delta')) = \emptyset$. This cannot happen as by the bijective relation between augmented tree nodes $V(\hatT_1^f(\delta'))$ and $V(\hatT_2^g(\delta'))$, this would have implied that $\levelC(S(\Intdopt)) \neq \emptyset$ but $\levelC(w(\Intdopt)) = \emptyset$. That is, $F(S, w) = 0$ in algorithm \DPalg($T_1^f,T_2^g, \Intdopt$) as well, which contradicts the assumption on $(S, w)$.  
\item[(2.b)]
Otherwise,  since $(S, w)$ is from the lowest level s.t. $F_{\Intdopt}(S,w)=1$ (w.r.t. parameter $\Intdopt$), yet $F_{\delta'}(S(\delta'),w(\delta')) = 0$ (w.r.t. parameter $\delta'$), condition (F-2) must have failed when we process $(S(\delta'), w(\delta')$ in algorithm \DPalg($T_1^f,T_2^g, \delta'$). In other words, there exists a child $w_j$ of $w$ such that $\mydepth(w_j) = 2\Intdopt - (g(w) - g(w_j))$, as condition (F-2) will be violated as we decrease $\Intdopt$ to $\delta'$. Thus there must be a leaf node $\hat{w}$ in the subtree rooted at $w_j$ so that $g(w_j) - g(\hat{w}) = \mydepth(w_j)$, and thus $g(w) - g(\hat{w}) = 2\Intdopt$. Assume that $S$ is from $\slone{i}$ (and thus $w$ is from $\sltwo{i}$). 
If \superlevel{} $\sltwo{i}$ contains a tree node of $T_2^g$, then this implies that $\Intdopt \in \criSet_3$. 
If not, then the \superlevel{} $\sltwo{i}$ must be an induced \superlevel{}, meaning that there is a tree node $u \in V(T_1^f)$ such that $g(w) = \hat{h}_i = f(u) + \Intdopt$. 
It then follows that $f(u) - g(\hat{w}) = \Intdopt$ and thus $\Intdopt \in \criSet_1$. 
\end{itemize}
\end{itemize}
In all cases, we note that $\Intdopt \in \criSet$. 
Lemma \ref{lem:criSet} then follows.

\section{Proof of Theorem \ref{thm:optinterleaving-fast}}
\label{appendix:thm:optinterleaving-fast}

First, we compute and sort all candidate values in the set $\criSet$ as in Section \ref{subsec:int-optimizaton}. 
Set $m = |\criSet|$ and let $\Delta = \{\delta_1 < \delta_2 < \cdots \delta_m \}$ be the ordered sequence of candidate values in $\criSet$. 

On the other hand, observe that the degree-bound parameter $\tau_\delta(T_1^f, T_2^g)$ can take at most $O(n)$ integer values: $1, 2, 3, \ldots, 2n$, where $n$ is the total number of nodes in the input trees. 
Set $\tau_k = k$, with $k \in \{1, 2, \ldots, 2n\}$. 
It is easy to see that the set of $\delta \in \Delta$ such that $\tau_\delta = \tau_k$ forms a consecutive subsequence of $\Delta$, which we denote by $[\delta_{\ell_k}, \ldots, \delta_{r_k}]$. That is, for any $j \in [\ell_k, r_k]$, $\tau_{\delta_j} = \tau_k$. 

For a fixed $\tau$, for simplicity, we set $Time(n, \tau):= n^2 2^\tau \tau^{\tau+2}\log n$; by Theorem \ref{thm:fasteralg}, $O(Time(n, \tau))$ is the time to solve the decision problem ``Is $d_I(T_1^f, T_2^g) \le \delta$ where $\delta$ is any value such that its $\delta$-degree-bound $\tau_\delta$ equal to $\tau$; that is, $\tau_\delta = \tau$. 

Below, we first show that for any $\tau_k$, we can identify the smallest $\delta \in \Delta$ \emph{valid for $\tau_k$} defined as (i) $\tau_\delta = \tau_k$, and (ii) $d_I(T_1^f, T_2^g) \le \delta$; or returns ``Null" if such a $\delta$ does not exist. 
We call this procedure {\sf SmallestValidDelta}($\tau_k$). 

\paragraph{Procedure {\sf SmallestValidDelta}($\tau_k$).}
First, we run the decision problem for $d_I(T_1^f, T_2^g) \le \delta_{r_k}$. If the answer is `no', then no $\delta$ can be valid for $\tau_k$, and we return ``Null". 
Otherwise, we perform a binary search within the range $[\ell_k, r_k]$ to search for the smallest index $i \in [\ell_k, r_k]$ such that $\delta_i$ is valid for $\tau_k$, and return $\delta_i$. 
Overall, we execute the decision problem $O(\log |r_k - \ell_k+1|) = O(\log n)$ times. For each of the $\delta$ value tested during this course, we have $\tau_\delta = \tau_k$. Hence by Theorem \ref{thm:fasteralg} the overall time complexity for this procedure is $O(Time(n, \tau) \cdot \log n)$. 

\paragraph{Identifying optimal $\delta^*$.} 
Now to prove Theorem \ref{thm:optinterleaving-fast}, we need to identify the smallest $\delta^* \in \Delta$ where the decision problem returns `yes'. 
Let $\tau_{k^*} = \tau_{\delta^*}$ be the $\delta^*$-degree bound parameter. 
If we can identify $k^*$, then Procedure {\sf SmallestValidDelta}($\tau_k$) will computes $\delta^*$. 

On the other hand, $k^*$ is the smallest $k$ value such that Procedure {\sf SmallestValidDelta}($\tau_k$) returns `yes'. We compute $k^*$ by a standard binary+exponential search procedure on $k \in \{1, 2, \ldots, 2n\}$. 
Specifically, starting with $k = 1$, we check whether Procedure {\sf SmallestValidDelta}($\tau_k$) returns ``Null" or not. If yes, then we double the size of $k$ (i.e, in the $i$-th iteration, we will be checking for $k = 2^i$), till we reach the first $k'$ such that the procedure returns a valid $\delta$ value. 
This means that the correct $k^*$ must be from $[k'/2, k']$, which we further identify by a binary search within this range. 

The time needed for the first ``exponential'' search of $k' = 2^{i'}$ is 
\begin{align*} 
& Time(n, 1) \log n + Time(n, 2) \log n + Time(n, 4) \log n + \cdots + Time(n, 2^{i'}) \log n \\
= &O(Time(n, k')\log n), 
\end{align*}
as the last term $Time(n, 2^{i'})\log n = Time(n, k') \log n$ dominates the sum of this  geometric-like series. 

For the second binary search within the range $[k'/2, k']$, it takes $O(Time(n, k') \log n \log k')$. As $k' \le 2k^*$, the total time is thus $O(Time(n, 2k^*) \log^2 n)$, as claimed in Theorem \ref{thm:optinterleaving-fast}. 

\section{Proof of Theorem \ref{thm:GHalg}}
\label{appendix:thm:GHalg}

By Claim \ref{claim:GHandInterleaving}, if we can compute $\mu = \min_{u\in V(T_1), w\in V(T_2)} d_I(T_1^{f_u}, T_2^{g_w})$, where $f_u: |T_1| \to \mathbb{R}$ (resp. $g_w: |T_2| \to \mathbb{R}$) is the distance function to point $u$ (resp. $w$) in $|T_1|$ (resp. in $|T_2|$), then $\mu$ is a $14$-approximation of $\hat{\Intdopt} = \dgh(\mathcal{T}_1,\mathcal{T}_2)$; that is, $\mu/14 \le \hat{\Intdopt} \le 14 \mu$ (the upper bound is in fact $2 \mu$ by Claim \ref{claim:GHandInterleaving}). 

There are two issues that we need to address. 
First, consider an arbitrary pair of nodes  $u\in V(T_1)$ and $w\in V(T_2)$. Our DP-algorithm to compute $d_I(T_1^{f_u}, T_2^{g_w})$ depends on the degree-bound $\tau_\delta(T_1^{f_u}, T_2^{g_w})$ parameter for the two merge trees $T_1^{f_u}$ and $T_2^{g_w}$. 
How does this relate to the metric-degree-bound $\newtau_\delta(T_1, T_2)$? We claim: 
\begin{claim}\label{claim:twodegreebounds}
$\newtau_\delta \le \tau_\delta \le \newtau_{2\delta}$.
\end{claim}
%We claim that $\newtau_\delta \le \tau_\delta \le 2 \newtau_\delta$. 
\begin{proof}
It is easy to see that for any $x\in |T_1|$, $\widehat{B}_\delta (x, T_1) \subseteq B_\delta(x, T_1^{f_u})$, where $\widehat{B}_\delta$ is the geodesic ball w.r.t. metric $d_1$, while $B_\delta$ 
is the $\delta$-ball defined at the beginning of Section \ref{sec:decision} for the merge tree $T_1^f$. 
A symmetric statement hold for the two type of balls in $T_2$ and $T_2^{g_w}$. 
This implies that $\newtau_\delta \le \tau_\delta$. 
We now show that for each $B_\delta(x, T_1^{f_u})$, there exists some $y\in |T_1|$ such that $B_\delta(x, T_1^{f_u}) \subseteq \widehat{B}_{2\delta}(y, T_1)$. 

Indeed, consider any $B_\delta(x, T_1^{f_u})$: First, observe that by  definition of the $\delta$-ball $B_\delta(x, T_1^{f_u})$ (see the beginning of Section \ref{sec:decision}), any point $z\in B_\delta(x, T_1^{f_u})$ is connected to $x$ and $f_u(z) \in [f_u(x) - \delta, f_u(x) + \delta]$. 
Set $y$ be the point this ball with highest $f_u$ function value. As $B_\delta(x, T_1^{f_u})$ is a connected subset of a tree, there is a monotone path $\pi(y, z)$ from this highest point $y$ in it to every point $z\in B_\delta(x, T_1^{f_u})$. 
Since $f_u$ is the shortest path distance to point $u \in T_1$ (which is the root of the merge tree $T_1^{f_u}$), and $\pi(y,z)$ is monotone, we thus have that $f_u(z) = f_u(y) + d_1(z, y)$. 
It then follows that $d_1(z,y) = f_u(y) - f_u(z) \le 2\delta$ for any point $z\in B_\delta(x, T_1^{f_u})$. Hence $B_\delta(x, T_1^{f_u}) \subseteq \widehat{B}_{2\delta}(y, T_1)$ as claimed. 

A symmetric statement holds for balls for $T_2$ and $T_2^{g_w}$. Hence $\tau_\delta \le \newtau_{2\delta}$. 
\end{proof}

Second, observe that in general, $d_I(T_1^{f_u}, T_2^{g_w})$ could be much larger than $\mu$. This means that the time complexity of our DP algorithm to compute $d_I(T_1^{f_u}, T_2^{g_w})$ may not be dependent on $\tau_{\mu}$ (and thus $\newtau_{2\mu}$) any more. 
Hence to compute $\mu = \min_{u\in V(T_1), w\in V(T_2)} d_I(T_1^{f_u}, T_2^{g_w})$, we cannot afford to compute every $d_I(T_1^{f_u}, T_2^{g_w})$ for each pair of $u\in V(T_1)$ and $w\in V(T_2)$. 
Instead, we now enumerate all candidate sets $\widehat{\criSet} = \cup_{u \in V(T_1), w\in V(T_2)} \criSet_{u, w}$, where $\criSet_{u,w}$ denotes the candidate set of choices for $\delta$ w.r.t $T_1^{f_u}$ and $T_2^{g_w}$, for all $O(n^2)$ pairs of $u \in V(T_1)$ and $w\in V(T_2)$.
We next sort all parameters $\delta_{u, w}$ in $\widehat{\criSet}$: the subscript $u, w$ in $\delta_{u,w}$ means that this value is a potential $\delta$ value for interleaving distance $d_I(T_1^{f_u}, T_2^{g_w})$. Note that The total number of critical values in $\widehat{\criSet}$ is $O(n^2 \times n^2) = O(n^4)$. This takes $O(n^4 \log n)$ time. 

Next, we perform a similar double-binary search procedure (once for the optimal degree-bound parameter $\tau$, and once for $\delta$) as in the proof of Theorem \ref{thm:optinterleaving-fast} to identify the smallest $\delta_{u,w}$ such that $d_I(T_1^{f_u}, T_2^{g_w}) \le \delta_{u,w}$ for any $u \in V(T_1)$ and $w\in V(T_2)$. 
Let $\tau' = \tau_{\delta_{u,w}} (T_1^{f_u}, T_2^{g_w})$ be the $\delta_{u,w}$-degree bound for merge trees $f_u$ and $g_w$. 
This means that during the procedure, all the parameter $\tau$ values we ever tested are at most $2\tau'$, meaning that for each decision problem we ever tested, its time complexity is bounded by $O(Time(n, 2\tau') \log^2 n)$, where $Time(n, \tau) := n^2 2^\tau \tau^{\tau+2}\log n$ as in the proof of Theorem \ref{thm:optinterleaving-fast}. 
Hence the total time complexity of this double-binary search procedure is 
$O(Time(n, 2\tau') \log^2 n)$. 

On the other hand, by Claim \ref{claim:GHandInterleaving}, we have $\delta_{u,w} (= \min_{u\in V(T_1), w\in V(T_2)} d_I(T_1^{f_u}, T_2^{g_w})) \le 14 \hat{\Intdopt}$. It then follows from Claim \ref{claim:twodegreebounds} that 
$$ \tau' = \tau_{\delta_{u,w}} \le \tau_{14 \hat{\Intdopt}} \le \newtau_{28 \hat{\Intdopt}}. $$
Setting $\newtau = 2\newtau_{28 \hat{\Intdopt}}$, the total time for performing the double-binary search procedure is 
$$ O(Time(n, 2\tau') \log^2 n) = O(n^2 2^{\newtau} \newtau^{\newtau+2} \log^3 n),  $$ 
which, together the $O(n^4\log n)$ time to compute and sort all candidate values in $\widehat{\criSet}$, gives the time complexity as claimed in Theorem \ref{thm:GHalg}.

% %%%%%%%%%%%%%%%%%% old $n^6$ time algorithm%%%%%%%%%%%%%%%%%%%%%%%%%%%%%%%%%%%
% We then consider the parameter 
% $\delta$ from $\widehat{\criSet}$ in increasing order: For each $\delta = \delta_{u,w}$, we perform our DP algorithm to check if $d_I(T_1^{f_u}, T_2^{g_w}) \le \delta_{u,w}$ holds. We stop the first time the answer to the decision problem is a `yes', and return the present $\delta_{u,w}$ value. It is easy to see that this $\delta_{u,w} = \min_{u\in V(T_1), w\in V(T_2)} d_I(T_1^{f_u}, T_2^{g_w}) (= \mu)$, which $14$-approximates $\delta^*= \delta_{GH}(\Tcal_1, \Tcal_2)$. 

% Since each time we run our DP algorithm, the parameter $\delta = \delta_{u,w} \le \mu \le 14 \hat{\Intdopt}$. By Claim \ref{claim:twodegreebounds}, the degree-bound $\tau_\delta$ satisfies $$\tau_\delta \le \tau_\mu \le \tau_{14\delta^*} \le  \newtau_{28\hat{\Intdopt}}.$$
% Hence the time to perform \DPalg($T_1^{f_u}, T_2^{g_w}, \delta_{u,w})$ is bounded by $O(n^2 2^\newtau \newtau^{\newtau+2}\log n)$ (by Theorem \ref{thm:fasteralg}), where $\newtau = \newtau_{28\hat{\Intdopt}}$. 
% As $|\widehat{\criSet}| = O(n^4)$, the time complexity to compute $\mu$ using the above procedure is thus $O(n^6 2^\newtau \newtau^{\newtau+2}\log n)$, as claimed. 
% This finishes the proof of Theorem \ref{thm:GHalg}. 

%\bibliographystyle{abbrv}
%\bibliography{ref}

\end{document}